\newif \iffullpaper \fullpapertrue

\newif \ifvone \vonetrue

\iffullpaper
\documentclass[sigconf,authorversion,screen]{acmart}
\else
\documentclass[sigconf]{acmart}
\fi

\copyrightyear{2024}
\acmYear{2024}
\setcopyright{acmlicensed}
\acmConference[CCS '24] {Proceedings of the 2024 ACM SIGSAC Conference on Computer and Communications Security}{October 14--18, 2024}{Salt Lake City, UT, USA.}
\acmBooktitle{Proceedings of the 2024 ACM SIGSAC Conference on Computer and Communications Security (CCS '24), October 14--18, 2024, Salt Lake City, UT, USA}
\acmISBN{979-8-4007-0636-3/24/10}
\acmDOI{10.1145/3658644.3690356}

\usepackage{xcolor}
\usepackage[english]{babel}
\newif\ifshowcomment
\showcommentfalse
\ifshowcomment
\newcommand{\deepak}[1]{\textsf{\small \color{teal}{[Deepak: {#1}]}}}
\newcommand{\mahdi}[1]{\textsf{\small \color{blue}{[Mahdi: {#1}]}}}
\newcommand{\arnab}[1]{\textsf{\small \color{red}{[Arnab: {#1}]}}}
\newcommand{\foteini}[1]{\textsf{\small \color{violet}{[Foteini: {#1}]}}}
\newcommand{\kostas}[1]{\textsf{\small \color{orange}{[kostas: {#1}]}}}
\newcommand{\reviewer}[1]{\textsf{\small \color{olive}{[Reviewer: {#1}]}}}
\else
\newcommand{\deepak}[1]{}
\newcommand{\mahdi}[1]{}
\newcommand{\arnab}[1]{}
\newcommand{\foteini}[1]{}
\newcommand{\kostas}[1]{}
\newcommand{\reviewer}[1]{}
\fi

\usepackage[utf8]{inputenc}
\usepackage{xcolor}
\usepackage{amsmath}
\usepackage{amsthm}
\usepackage{mathtools}
\usepackage{float}
\usepackage{amsfonts}
\usepackage{comment}
\usepackage{listings} 
\usepackage{hyperref}
\newcommand{\ra}[1]{\renewcommand{\arraystretch}{#1}}
\hypersetup{
	linktoc = page,
	pdfpagemode = UseNone,
	colorlinks,
	linkcolor={red!50!black},
	citecolor={blue!50!black},
	urlcolor={blue!80!black}
}

\usepackage[algo2e,linesnumbered,ruled,vlined]{algorithm2e}
\usepackage{algorithm}
\usepackage{pifont}
\usepackage{tabularx}
\usepackage{tikz}
\newcommand*\circled[1]{\tikz[baseline=(char.base)]{
            \node[shape=circle,draw,inner sep=0.3pt] (char) {#1};}}

\usetikzlibrary{shapes.geometric, positioning, arrows}

\usepackage{pbox}

\definecolor{grey1}{rgb}{0.65,0.65,0.65}

\usepackage{dashbox}

\usepackage{booktabs}
\newcommand{\secparam}{\lambda}
\newcommand{\secpar}{\secparam}
\usepackage{verbatim}
\usepackage[strings]{underscore}

\usepackage{caption}
\usepackage{subcaption}

\usepackage{thmtools}
\usepackage{thm-restate}
\usepackage{url}

\usepackage{breakurl}
\usepackage{enumitem}
\setlist[itemize]{itemsep=1.5pt, topsep=2pt}
\setlist[enumerate]{itemsep=1.5pt, topsep=2pt}
\setlist[description]{itemsep=1.5pt, topsep=2pt}

\renewcommand{\paragraph}[1]{\vspace{4pt} \noindent \textbf{#1:}}

\DeclareCaptionLabelFormat{cont}{#1~#2\alph{ContinuedFloat}}
\renewcommand{\sim}{\mathsf{Sim}}

\newcommand{\advA}{\mathcal{A}}

\newcommand{\setup}{\mathsf{Gen}}
\newcommand{\keygen}{\mathsf{KeyGen}}

\newcommand{\sk}{\mathsf{sk}}
\newcommand{\pk}{\mathsf{pk}}

\newcommand{\sign}{\mathsf{Sign}}
\newcommand{\verify}{\mathsf{Verify}}

\newcommand{\sample}{\overset{\$}{\leftarrow}}

\newcommand{\vk}{\mathsf{vk}}

\newcommand{\extract}{\mathsf{Ext}}

\newcommand{\ignore}[1]{}

\newcommand{\oracle}{\mathcal{O}}
\newcommand{\pcreturn}{\mathsf{return}}

\newcommand{\bin}{\{0,1\}}
\newcommand{\bits}{\bin}

\newcommand{\ext}{\mathsf{Ext}}

\newcommand{\crs}{\mathsf{CRS}}
\let\negl\relax
\newcommand{\negl}{\mathsf{negl}}

\usepackage[capitalise,compress,nameinlink]{cleveref}
\Crefname{equation}{Eq.}{Eqs.}
\Crefname{figure}{Fig.}{Figs.}
\Crefname{tabular}{Tab.}{Tabs.}
\Crefname{appendix}{App.}{Apps.}
\Crefname{section}{Sec.}{Sects.}
\Crefname{definition}{Def.}{Defs.}

\newcommand{\appref}[1]{
    \iffullpaper
        \cref{#1}
    \else
        the full paper~\cite{zkLogin-arXiv}
    \fi
}

\newcommand{\prove}{\mathsf{Prove}}

\newcommand{\zk}{\mathsf{ZK}}

\definecolor{irrelgray}{gray}{0.6}

\colorlet{punct}{red!60!black}
\definecolor{background}{HTML}{EEEEEE}
\definecolor{delim}{RGB}{20,105,176}
\colorlet{numb}{magenta!60!black}
\colorlet{comment}{green!70!black}

\lstdefinestyle{interfaces}{
  float=tp
}

\lstdefinelanguage{json}{
    basicstyle=\scriptsize\ttfamily,
    stepnumber=1,
    numbersep=8pt,
    showstringspaces=false,
    breaklines=true,
    frame=lines,
    backgroundcolor=\color{background},
    literate=
     *{0}{{{\color{numb}0}}}{1}
      {1}{{{\color{numb}1}}}{1}
      {2}{{{\color{numb}2}}}{1}
      {3}{{{\color{numb}3}}}{1}
      {4}{{{\color{numb}4}}}{1}
      {5}{{{\color{numb}5}}}{1}
      {6}{{{\color{numb}6}}}{1}
      {7}{{{\color{numb}7}}}{1}
      {8}{{{\color{numb}8}}}{1}
      {9}{{{\color{numb}9}}}{1}
      {:}{{{\color{punct}{:}}}}{1}
      {,}{{{\color{punct}{,}}}}{1}
      {\{}{{{\color{delim}{\{}}}}{1}
      {\}}{{{\color{delim}{\}}}}}{1}
      {[}{{{\color{delim}{[}}}}{1}
      {]}{{{\color{delim}{]}}}}{1},
    morestring=[b]",
    commentstyle=\color{comment},
    morecomment=[l]{\#} 
}
\lstdefinestyle{jwt}{
    backgroundcolor=\color{background},   
    basicstyle=\sffamily\footnotesize,
    commentstyle=\color{comment},
    frame=lines,
    morecomment=[l]{\#} 
}

\newcommand{\Wit}{\mathsf{Wit}}

\def\TWS{Tagged Witness Signature\xspace}

\newcommand{\sub}{\mathsf{sub}}
\newcommand{\aud}{\mathsf{aud}}
\newcommand{\iss}{\mathsf{iss}}
\newcommand{\jwt}{\mathsf{jwt}}
\newcommand{\sfemail}{\mathsf{email}}
\newcommand{\issueJWT}{JWT.\mathsf{Issue}}
\newcommand{\verifyJWT}{JWT.\mathsf{Verify}}
\newcommand{\nonce}{\mathsf{nonce}}
\newcommand{\claimset}{\mathcal{C}}
\newcommand{\stableid}{\mathsf{stid}}

\renewcommand{\sk}[1]{\mathsf{sk}_{#1}}

\def\sysname{zkLogin\xspace}
\def\systws{\Sigma_{\mathsf{\sysname}}}

\newcommand{\ephsk}{sk_{u}}
\newcommand{\ephpk}{vk_{u}}

\newcommand{\OPsk}{\sk{OP}}
\newcommand{\OPpk}{\pk_{OP}}

\newcommand{\zkaddr}{\mathsf{zkaddr}}
\newcommand{\salt}{\mathsf{salt}}
\newcommand{\saltSeed}{k_{seed}}

\newcommand{\expiryTime}{T_{max}}
\newcommand{\curTime}{T_{cur}}
\newcommand{\expiryDelta}{\delta}
\newcommand{\nonceRand}{r}

\renewcommand{\tag}{\mathsf{tag}}

\newcommand{\curOPpk}{\pk_{OP}}
\newcommand{\curOPsk}{\sk{OP}}

\newcommand{\td}{\mathsf{trap}}
\newcommand{\SimGen}{\mathsf{SimGen}}
\newcommand{\simsign}{\mathsf{SimSign}}

\newcommand{\pzk}{\mathsf{Ckt}}
\newcommand{\zkx}{\mathsf{zkx}}
\newcommand{\zkw}{\mathsf{zkw}}

\usepackage{varwidth} 
\RequirePackage[most]{tcolorbox}

\newtheorem{definition}{Definition}
\newtheorem{theorem}{Theorem}

\newtcolorbox{titlebox}[5]{enhanced,center,colframe=black,colback=white,boxrule={#3},arc={#2},auto outer arc,%
	breakable,pad at break*=5pt,vfill before first,before={\par\medskip\noindent},after={\par\medskip},top=12pt,left=4pt,%
	enlarge top by=7pt,
	title={\rule[-.3\baselineskip]{0pt}{\baselineskip}\normalsize\sffamily\bfseries #1}, varwidth boxed title*=-30pt, 
	attach boxed title to top left={yshift=-10pt,xshift=10pt}, coltitle=black,
	boxed title style={colback=white,boxrule={#5},arc={#4},auto outer arc}
}

\usepackage{setspace}

\newsavebox{\fboxenvbox}
\newenvironment{fboxenv}
{\begin{lrbox}{\fboxenvbox}}
	{\end{lrbox}\fbox{\usebox{\fboxenvbox}}}

\newcommand{\uline}[1]{\underline{#1}}

\makeatletter
\newcommand{\subalign}[1]{%
	\vcenter{%
		\Let@ \restore@math@cr \default@tag
		\baselineskip\fontdimen10 \scriptfont\tw@
		\advance\baselineskip\fontdimen12 \scriptfont\tw@
		\lineskip\thr@@\fontdimen8 \scriptfont\thr@@
		\lineskiplimit\lineskip
		\ialign{\hfil$\m@th\scriptstyle##$&$\m@th\scriptstyle{}##$\hfil\crcr
			#1\crcr
		}%
	}%
}
\makeatother

\newcommand{\Gen}{\mathsf{Gen}}

\newcommand{\comm}{\mathsf{Com}}

\newcommand{\twsuf}{\mathsf{EUF\mbox{-}CTMA}}

\title{\sysname: Privacy-Preserving Blockchain Authentication with Existing Credentials}

\begin{document}
\author{Foteini Baldimtsi}
\orcid{0000-0003-3296-5336}
\affiliation{
    \institution{Mysten Labs}
    \city{Palo Alto} 
    \state{CA}
    \country{USA}
}
\affiliation{
    \institution{George Mason University}
    \city{Fairfax}
    \state{VA}
    \country{USA}
}
\email{foteini@mystenlabs.com}

\author{Konstantinos Kryptos Chalkias}
\orcid{0000-0002-3252-9975}
\affiliation{
    \institution{Mysten Labs} 
    \city{Palo Alto} 
    \state{CA}
    \country{USA}
}
\email{kostas@mystenlabs.com}
\authornotemark[1]

\author{Yan Ji}
\orcid{0000-0002-9448-2164}
\affiliation{
    \institution{Cornell Tech}
    \city{New York}
    \state{NY}
    \country{USA}
}
\email{jyamy42@gmail.com}

\author{Jonas Lindstrøm}
\orcid{0000-0002-1989-3019}
\affiliation{
    \institution{Mysten Labs} 
    \city{Palo Alto} 
    \state{CA}
    \country{USA}
}
\email{jonas@mystenlabs.com}

\author{Deepak Maram}
\orcid{0000-0001-5324-6889}
\affiliation{
    \institution{Mysten Labs} 
    \city{Palo Alto} 
    \state{CA}
    \country{USA}
}
\email{deepak@mystenlabs.com}
\authornote{Corresponding authors}

\author{Ben Riva}
\orcid{0009-0008-2192-1905}
\affiliation{
    \institution{Mysten Labs} 
    \city{Palo Alto} 
    \state{CA}
    \country{USA}
}
\email{benriva@mystenlabs.com}

\author{Arnab Roy}
\orcid{0009-0005-3770-9982}
\affiliation{
    \institution{Mysten Labs} 
    \city{Palo Alto}
    \state{CA}
    \country{USA}
}
\email{arnab@mystenlabs.com}

\author{Mahdi Sedaghat}
\orcid{0000-0002-1507-6927}
\affiliation{
    \institution{COSIC, KU Leuven} 
    \city{Leuven}
    \country{Belgium}
}
\email{ssedagha@esat.kuleuven.be}

\author{Joy Wang}
\orcid{0009-0007-9002-3191}
\affiliation{
    \institution{Mysten Labs} 
    \city{Palo Alto} 
    \state{CA}
    \country{USA}
}
\email{joy@mystenlabs.com}

\renewcommand{\shortauthors}{Foteini Baldimtsi et al.}

\begin{abstract}
    For many users, a private key based wallet serves as the primary entry point to blockchains. Commonly recommended wallet authentication methods, such as mnemonics or hardware wallets, can be cumbersome. This difficulty in user onboarding has significantly hindered the adoption of blockchain-based applications.

We develop zkLogin, a novel technique that leverages identity tokens issued by popular platforms (any OpenID Connect enabled platform e.g., Google, Facebook, etc.) to authenticate transactions. At the heart of zkLogin lies a signature scheme allowing the signer to \textit{sign using their existing OpenID accounts} and nothing else. This improves the user experience significantly as users do not need to remember a new secret and can reuse their existing accounts.

zkLogin provides strong security and privacy guarantees. Unlike prior works, zkLogin's security relies solely on the underlying platform's authentication mechanism without the need for any additional trusted parties (e.g., trusted hardware or oracles). As the name suggests, zkLogin leverages zero-knowledge proofs (ZKP) to ensure that the sensitive link between a user's off-chain and on-chain identities is hidden, even from the platform itself. 

zkLogin enables a number of important applications outside blockchains. It allows billions of users to produce \textit{verifiable digital content leveraging their existing digital identities}, e.g., email address. For example, a journalist can use zkLogin to sign a news article with their email address, allowing verification of the article's authorship by any party.

We have implemented and deployed zkLogin on the Sui blockchain as an additional alternative to traditional digital signature-based addresses. Due to the ease of web3 on-boarding just with social login, many hundreds of thousands of zkLogin accounts have already been generated in various industries such as gaming, DeFi, direct payments, NFT collections, sports racing, cultural heritage, and many more.
\end{abstract}



\maketitle

\renewcommand{\state}{\mathsf{st}}

\section{Introduction} \label{sec:introduction}





Blockchains are decentralized ledgers maintained by a network of validators or miners. 
The blockchain ledger functions as an append-only record, logging transactions in a secure and immutable manner. 
In existing designs, each user is equipped with a unique pair of cryptographic keys: a private key and a public key. The private key of a user essentially holds the user's assets and is used to execute transactions.
To initiate a transaction, a user digitally signs it using their private key, and validators can confirm the validity of the signed transaction using the corresponding public key.
Once verified, transactions are permanently added to the blockchain.




Users can opt to store their blockchain secret keys in a self-managed, or else non-custodial wallet. While this option gives full control to users, it also comes with the responsibility to store, manage, and secure their private keys. 
If a private key is lost, the associated assets are no longer retrievable.  
For example, in the case of Bitcoin, it is estimated that 7\% of all coins are lost forever~\cite{lostbitcoins2024}.
A natural solution would instead be to resort to custodial services.
While these platforms offer a more intuitive user experience reminiscent of traditional online platforms, their reliability is contentious.
The downfall of notable custodial firms~\cite{huang2022ftxcollapse, harkin2021mtgox, morris2021gerald, rizzo2017btc}, whether due to mismanagement, security hacks or fraud, has made it difficult for users to place faith in emerging entities~\cite{yu2024don}.

A potential resolution to this predicament is to leverage the existing trust that users have in globally recognized platforms, e.g. Google, Apple etc.
The ubiquity and acceptance of standards like OAuth 2.0, which allow for the use of an existing account from one platform to authenticate on another, could serve as a direct gateway for integrating users of their platforms into the blockchain ecosystem.



However, a direct use of OAuth requires the introduction of a new trusted party for authentication purposes. 
Specifically, the OAuth protocol allows an OAuth Provider (e.g., Google) to convince an OAuth client (either a server or a piece of front-end code) about user-specific details (e.g., email).
However, since a blockchain cannot function as an OAuth client, this model would necessitate the introduction of a trusted web server, functioning as an oracle, to relay pertinent information to the blockchain.\footnote{Alternatively, the user could run separate instances of OAuth, one with each validator. However this is cumbersome and impractical.}


This scenario naturally leads to a pivotal question: Can we harness existing authentication systems to oversee a cryptocurrency wallet, \emph{without necessitating reliance on additional trusted entities}?

We answer the above question in the affirmative.
Our approach relies on the OpenID Connect specification~\cite{openid-connect-core}, that is commonly conformed to by the prevalent OAuth 2.0 providers.
OpenID providers (OP) issue a signed statement, referred to as a JSON Web Token (JWT).
An example JWT with a dummy payload is in~\cref{fig:jwt}.
The JWT's payload contains basic user information, as shown in the more realistic example payload of~\cref{lst:jwt-payload}.



\begin{figure}[ht!]
  \centering
\includegraphics[scale=0.5]{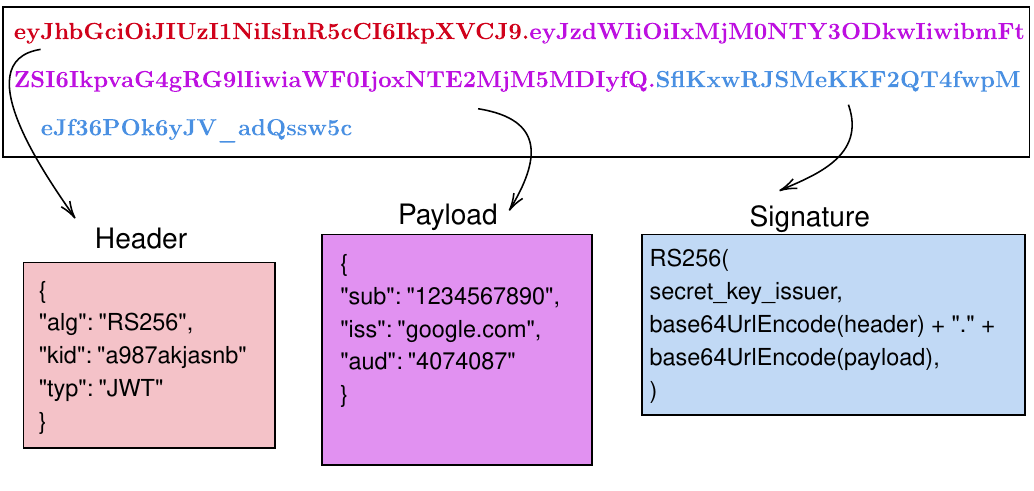}
  \caption{JSON Web Token.}
  \label{fig:jwt}
\end{figure}

\begin{lstlisting}[language=json, label={lst:jwt-payload}, caption=JWT Payload.,float,captionpos=b]
{
  "sub": "1234567890", # User ID
  "iss": "google.com", # Issuer ID
  "aud": "4074087", # Client or App ID
  "iat": 1676415809, # Issuance time
  "exp": 1676419409, # Expiry time
  "name": "John Doe",
  "email": "john.doe@gmail.com",
  "nonce": "7-VU9fuWeWtgDLHmVJ2UtRrine8"
}
\end{lstlisting}

The main idea in \sysname is to utilize the JWT's signature to directly authenticate the user with a blockchain, thus eliminating the need for any middlemen.

A strawman way of realizing this would be as follows. 

\begin{enumerate}
    \item The user logs in to their existing OP account (say on Google), leveraging OpenID Connect to obtain a JWT.
    \item The JWT is sent on the blockchain, e.g., to a contract.
    \item The embedded signature within the JWT facilitates its verification. Note that the contract would need to store the public keys of the said OpenID provider to be able to verify the JWT.
    \item The contract can employ the persistent subject identifier ($\sub$) present in the JWT to be able to identify the same user across different sessions.
\end{enumerate}


A similar approach was previously proposed in~\cite{OPENZEPPELIN}.
While this can work, the main problem is that it reveals the entire JWT payload publicly, including sensitive claims such as name, email, profile image, etc.
This is very problematic in the case of public blockchains (the focus of our work) like Ethereum, Solana or Sui where the state of a blockchain is completely public.
The above solution only focuses on authentication, not showing how the user can authorize a blockchain transaction, which is needed to truly realize a wallet.

\subsection{The \sysname Approach}
A natural way to avoid revealing the entire JWT is to leverage Zero-Knowledge Proofs (ZKP)~\cite{STOC:GolMicRac85}.
In particular, the user can input the JWT, which we define by $J$, as a private witness and prove that $J$ contains a valid signature issued by the OpenID provider among other things. 


Although the use of ZKP could in principle solve the privacy concerns, multiple challenges arise if one attempts to realize the above idea in a practical and compatible manner.

First, existing JWTs use traditional cryptographic primitives like SHA-2 and RSA signatures which are not ZK-friendly. In addition, most existing state-of-the-art ZK Proving systems incur high computational-overhead for proving (focusing more on reducing the verification complexity). 
This state of affairs implies that we need to employ powerful hardware to be able to generate proofs efficiently.
But in our setting, the proving entity is the user -- which means that the ZKP may need to be generated in resource-constrained environments, e.g., poor hardware / browsers, thus making it impractical for many users today.

Moreover, na\"ive approaches require generating a new ZKP for every transaction that the user signs (e.g., see recent work~\cite{EPRINT:PLLCCKCM23}).
    This further compounds the previous issue.



A simple trick helps us overcome the above challenges.
Before getting a JWT, the user generates an ephemeral key pair~$(\ephsk, \ephpk)$ and implants the public key~$\ephpk$ into the nonce during the OpenID Connect protocol (OpenID uses nonce to prevent replay attacks).
The signed JWT~$J$ thus acts as a certificate for the ephemeral public key, and we can reuse the corresponding private key~$\ephsk$ to sign any number of transactions.
%
Implanting public key into the nonce allows authorizing transactions.
The \sysname-signature on a transaction~$tx$ contains two steps: 
\begin{enumerate}
    \item a ZKP proving validity of a JWT~$J$ and showing that it contains~$\ephpk$ in its nonce, and
    \item a traditional digital signature on~$tx$ with~$\ephsk$.
\end{enumerate}

Notice that a single ZKP can be reused to sign \textit{any number of transactions}~-- thus amortizing the cost of the expensive ZKP generation.
The ephemeral key pair can be deleted at an appropriate time, e.g., after a browsing session ends.

While the above can help reduce the number of times a ZKP needs to be generated, the user still needs to generate proofs once in a while.
This may not always be practical (today) as we find that proving moderately complex ZKPs (e.g., around 1M constraints in Groth16) can lead to crashes or long delays on a browser (we have not tested on mobile or desktop environments where we suspect local ZKP generation might be more feasible).

Therefore, we provide an option to \textit{offload the proof generation} to a different server in a way that this entity cannot create complete \sysname signatures on its own, since it will not never learn the ephemeral private key $\ephsk$.
Essentially, we can offload the Zero-Knowledge Proof generation (first step above) to a server, and once the server returns the proof, the user verifies it efficiently and completes the \sysname signature locally (second step).

To summarize, the idea of embedding data into the nonce helps us solve three challenges, namely, (a) authorize transactions, (b) reuse a single ZKP across many transactions, and (c) offload ZKP generation securely, if needed.


\paragraph{Identifying the user on-chain}
While using a ZKP can hide most of the sensitive information in a JWT, one more challenge remains. 
Any authentication system needs a way to persistently identify a user across sessions.
In today's private-key based wallets, this role is  neatly fulfilled by the public key which gets used to derive a user's blockchain address.
In \sysname, a unique and persistent user identifier from the JWT can be used to generate a user address.
We call such an identifier as a ``stable identifier''.
A few possible options for a stable identifier include the subject identifier ($\sub$), email address or username.


\subsection{\sysname Features}

Using a widely used identifier like email (or username) as the stable identifier makes the \sysname account \textit{easily discoverable}.
This can be useful for entities wanting to maintain a public blockchain profile for transparency reasons~\cite{techreview2023cryptography}, e.g., a journalist could digitally sign a news article using their existing email address or a photographer may sign a photo using their existing Facebook account.
Prior to \sysname, this was only possible through the use of trusted oracles to port legacy credentials~\cite{candid2021, zhang2020deco}.


Discoverability, however, comes with an inherent privacy problem as the link between the user's stable identifier and a blockchain address is forever public.

Therefore, we do not make \sysname accounts discoverable by default.
Instead, we use an additional randomizer in the form of a ``salt'' to hide the user's off-chain identity.
A user's address is a hash of the stable identifier, salt and a few other fields (e.g., the OpenID Provider's and the application's unique IDs).
Without knowledge of the salt, no entity can link a \sysname address to its corresponding off-chain identity.
We refer to this property as \textit{unlinkability}.

Note that in both cases, we achieve unlinkability from all entities including the OP (except the app in the first case).
A key consideration is who manages the salt (\cref{sec:zklogin-modes}): either the application (no unlinkability from the app but simpler UX) or the user (more complex UX in some cases but unlinkability from all parties).


Another feature of \sysname is its ability to create \textit{anonymous blockchain accounts}.
The user can hide sensitive parts of their stable identifier (like email), effectively leading to a \textit{ring signature}.
For example, use only the domain of an email address as the user's identity (the domain of the email ``ram@example.com'' is ``example.com'').
This approach is suitable in settings where both discoverability and anonymity are desired, e.g., attesting to the individual's affiliation with a specific organization, like a news outlet or educational institution or a country, while maintaining their anonymity.

\sysname can also be used to create \textit{targeted claimable accounts}, i.e., safely sending assets to a specific target user even before they have a blockchain account.
A sender can derive the receiver's \sysname address using the receiver's email address and a randomly chosen salt.
The salt can be sent to the receiver over a personal channel, e.g., using an E2E encrypted chat.
Like before, the receiver can choose to manage the newly received salt by themselves or delegate its management to an app.
To the best of our knowledge, creating targeted claimable accounts was not possible before (without revealing the receiver's private key to the sender which is undesirable).





\subsection{Technical Challenges}




\paragraph{Expiring the ephemeral keys}
In practice, it is prudent to set a short expiry time for the ephemeral key pair~$\ephpk$ for security reasons. 
A first idea is to use the JWT expiry time, e.g., the ``exp'' claim in~\cref{lst:jwt-payload}.
However, this is not ideal because applications may want more control over its expiry, e.g., many JWTs expire 1hr after issuance which may be too small.
Moreover, it is challenging to use real time in blockchains that skip consensus for certain transactions~\cite{blackshear2023sui}.

\sysname facilitates setting an arbitrary expiry time~$exp$ by embedding it into the nonce.
For example, if a blockchain publishes a block once every 10 mins and the current block number is $cur$, and we'd like to expire after ten hours (600 mins), then set~$exp = cur + 60$ and compute $nonce = H(\ephpk, exp)$.
Note that it is convenient to use the chain's local notion of time, e.g., block or epoch numbers, than the real time.
More broadly, arbitrary policy information governing the use of the ephemeral key~$\ephpk$ can be embedded into the nonce, e.g., permissions on what can be signed with~$\ephpk$.




\paragraph{Formalization}
\sysname closely resembles Signatures of Knowledge~\cite{C:ChaLys06} where the knowledge of a witness is enough to produce a valid signature.
The key difference in \sysname is that witnesses (JWT and ephemeral key pair) expire.

We capture this property 
by proposing a novel cryptographic primitive called \textit{\TWS}. The \TWS syntax provides an interface for a user to sign a message by demonstrating that it can obtain a secret witness, namely a JWT, to  a public ``tag", namely the OpenID Provider's public key.
A \TWS has two main properties: unforgeability and privacy. Unforgeability states that it is hard to adversarially forge a signature, even if witnesses to other tags get leaked (e.g., expired JWTs). 
Privacy states that it is hard for an adversary to learn non-public components of the witness, e.g., the JWT and the salt, from the signature. 

\paragraph{Implementation}
We instantiate the Zero-Knowledge Proof using Groth16~\cite{EC:Groth16} as the proving system and circom DSL~\cite{circom} as the circuit specification language (cf.~\cref{sec:implementation}). 

The main circuit operations are RSA signature verification and JWT parsing to read relevant claims, e.g., ``sub'', ``nonce''.
We use previously optimized circuits for RSA verification and write our own for JWT parsing.

\ifvone
Na\"ively done, parsing the JWT requires fully parsing the resulting JSON, which would've required implementing a complete JSON parser in R1CS.
\fi
We manage to optimize significantly by observing that the JSONs used in JWTs follow a much simpler grammar.
Observing that all the claims of interest, e.g., ``sub'' are simple JSON key-value pairs, we can \textit{only parse specific parts of the JWT}, namely the JSON key-value pairs of interest.


\ifvone
Our final circuit has around a million constraints.
SHA-2 is the most expensive taking 66\% of the constraints whereas RSA big integer operations take up 14\% of the constraints.
Thanks to above optimizations, the JWT parsing circuit only takes the remaining 20\% of the constraints, whereas a na\"ive implementation would've resulted in significantly more.
\fi

\subsection{Other Applications: Content Credentials}

The core primitive in \sysname can be viewed as an Identity-Based signature (IBS)~\cite{C:Shamir84}.
In an IBS, a key distribution authority issues a signing key over a user's identity~$id$, e.g., their email address, such that a user-generated signature can be verified using just their identity $id$, thereby eliminating the need for a Public Key Infrastructure (PKI).

\sysname can be viewed as an IBS where the \textit{OpenID provider implicitly functions as the key distribution authority} (\sysname also requires the existence of an app implementing the OAuth flows, which may be viewed as another component of the key distribution authority). 
This enables a number of critical applications. 
With the rise of generative AI, knowing the authenticity of content, e.g., emails, documents or text messages, has become challenging~\cite{techreview2023cryptography}.
A recent proposal by major technology firms attempts to establish provenance via \textit{content credentials}~\cite{c2pa2023website}, a cryptographic signature attached to a piece of digital content, e.g., news article, photos or videos.
Issuing content credentials requires setting up a new PKI whereas our IBS scheme facilitates creating content credentials without having to setup one from scratch.

\subsection{Contributions}\label{sec:contributions}

In summary, our contributions are as follows:

\begin{enumerate}
    \item We propose \sysname, a novel approach to the design of a blockchain wallet that offers significantly better user experience than traditional wallets, thanks to its use of well-established authentication methods. Moreover, \sysname offers novel features like discoverability and claimability that enable critical applications. 
    \item We introduce the notion of tagged witness signatures to formally capture the cryptographic core of \sysname, and prove its security.
    \item We implement \sysname using Groth16 as the NIZK in just around 1M R1CS constraints, thanks to several circuit optimizations, e.g., efficient JSON parsing, and string slicing, that maybe of independent interest. Generating a \sysname signature only takes about 3s.
\end{enumerate}

\paragraph{Structure of the Paper}
We start off the rest of the paper with an overview of OpenID in \cref{sec:preliminaries}. In \cref{sec:witsig}, we define \TWS along with its security and privacy properties. In \cref{sec:zklogindetails}, we describe the \sysname system. In \cref{sec:implementation}, we describe our production deployment of \sysname and document its performance. Finally, in \cref{sec:relatedworks}, we review existing works and conclude in \cref{sec:conclusion}.



\section{Preliminaries: OpenID Connect} \label{sec:preliminaries}




\label{sec:oidc}

OpenID Connect (OIDC) is a modern authentication protocol built on top of the OAuth 2.0 framework. 
It allows third-party applications to verify the identity of end users based on the authentication performed by an OpenID Provider (OP), e.g., Google, as well as to obtain basic profile information about the end user.
Not all OAuth 2.0 conforming providers implement OpenID Connect but most of the popular providers (Google, Apple, Facebook, Microsoft, etc.) do, which suffices for our purpose.
OIDC introduces the concept of an ID token, which is a JSON Web Token (JWT) that contains \textit{claims} about the authenticated user.

JSON Web Tokens (JWTs) are a versatile tool for securely transmitting information between parties using a compact and self-contained JSON format. A JWT consists of three components: a header, a payload, and a signature (see~\cref{fig:jwt}). All the three components are encoded in base64.
Decoding the header and payload results in JSON structures.
Sticking to JWT terminology, we refer to a JSON key as a \textit{claim name} and the corresponding value as the \textit{claim value}.


\paragraph{JWT header} \cref{fig:jwt} also shows a decoded JWT header. 
The ``alg'' claim specifies the signing algorithm used to create the signature.
The JSON Web Algorithms spec recommends the use of two algorithms: RS256 and ES256~\cite{rfc7518-section3} for this purpose.
Of the two, we found that RS256 is the most widely used, hence we only support that currently.  

The ``kid'' claim helps identify the key used for signature verification.
Let $(\OPsk, \OPpk)$ generate the actual key pair where there is a one-to-one mapping between the ``kid'' value and $\OPpk$.
The public key $\OPpk$ is posted at a public URI in the form of a JSON Web Key (JWK), e.g., Google posts its keys at \url{https://www.googleapis.com/oauth2/v3/certs}.
Moreover, many providers rotate keys frequently, e.g., once every few weeks~-- so a JWT verifier needs to periodically fetch the JWKs from the OP's website.

\paragraph{JWT Payload} \cref{lst:jwt-payload} shows an example JWT payload. Any OIDC-compliant JWT contains the following claims:

\begin{enumerate}
    \item ``iss'' (issuer): Identifies the entity that issued the JWT, typically an OpenID Provider (OP). This claim is fixed per OP, i.e., it's value is same for all tokens issued by the same OpenID Provider.
    \item ``sub'' (subject): Represents the subject of the JWT, often the user or entity the token pertains to. This claim is fixed per-user. The spec defines two approaches an OP can take to generate the subject identifier:
    \begin{enumerate}
        \item \textit{Public identifier}: Assign the same subject identifier across all apps. A majority of current providers choose public identifiers, e.g., Google, Twitch, Slack.
        \item \textit{Pairwise identifier}: Assign a unique subject identifier for each app, so that apps do not correlate the end-user's activities. E.g., Apple, Facebook, Microsoft.\footnote{The notion of a public (vs) pairwise identifier also applies to other identifiers, e.g., Apple's JWTs can include a pairwise email address, effectively allowing users to hide their real email addresses.}
    \end{enumerate}
    \item ``aud'' (audience): Defines the intended recipient(s) of the token, ensuring it's only used where it's meant to be. This claim is fixed per-app. The $\aud$ value is assigned to an app after it registers with the OP.
    \item ``nonce'': A unique value generated by the client to prevent replay attacks, particularly useful in authentication and authorization flows.
\end{enumerate}

Apart from the above, OIDC allows providers to include some optional claims like emails or set some custom claims.



\paragraph{JWT API} We model the process of issuing and verifying a JWT as follows:

\begin{itemize}
    \item $\jwt \gets \issueJWT(\OPsk, \claimset)$: After the user successfully authenticates, the OpenID Provider signs the claim set $\claimset = \{\sub, \aud, \iss, \nonce, \ldots\}$, and returns a Base64-encoded JWT as shown in~\cref{fig:jwt}.
    The OIDC spec mandates the presence of certain claims like $\sub, \aud, \iss, \nonce$\footnote{For nonce, the spec mandates that a nonce claim be present if the request contains it.} in the claim set. 
    
    \item $0/1 \gets \verifyJWT(\OPpk, \jwt)$: Verifies that the JWT was indeed signed by the OpenID Provider.
\end{itemize}

We use the notation $\jwt.{\sf claimName}$ to refer to the value of a particular claim in the JWT.
For example, if $\jwt$ refers to the example in~\cref{fig:jwt}, then $\jwt.\sub = ``1234567890"$ is the subject identifier.

\section{\TWS}
\label{sec:witsig}

In traditional digital signature schemes, a signer needs to maintain a long-lived secret key which can be burdensome. The goal of  
{\TWS}s (TWS) is to slightly relax this requirement by replacing the secret signing key with a valid \textit{witness} to a public statement.
The statement comprises of a tag $t$, with respect to a public predicate $P$, which is fixed for the scheme. At a high level, a TWS should satisfy the properties of completeness, unforgeability and witness-hiding.
The completeness property ensures that signatures produced with a valid witness $w$, i.e., $P(t, w) = 1$ verify. 
The unforgeability property ensures that an adversary cannot produce a valid signature without knowing a valid witness. 
The witness-hiding property guarantees that a signature does not reveal any information about the witness, essentially capturing zero-knowledge or privacy.


Like Signatures of Knowledge (SoK) \cite{C:ChaLys06}, there is no explicit secret key required for signing - the ``secret" is the ability to obtain a witness. However, a crucial difference is that unforgeability holds even if witnesses corresponding to different tags are leaked to the adversary. 
Modeling witness leakage is crucial in practical settings where the chances of an old witness leaking over a long-enough duration of time are high, such as in \sysname. In this sense, \TWS can be thought of as SoK with forward secrecy. We also achieve a few more desirable properties compared to Krawczyk~\cite{CCS:Krawczyk00} who defined forward secure signatures. In contrast to the construction of \cite{CCS:Krawczyk00}, our construction allows arbitrary numbers of dynamically defined timestamps. 



In addition, in contrast to SoK, we employ a $\Gen(\cdot)$ algorithm that preprocesses the predicate and generates specific public parameters for it. This design eliminates the need for the verifier to know or read the predicate, enabling significant optimizations in the blockchain environment. 


\begin{definition}[\TWS]
    A tagged witness signature scheme, for a predicate $P$, is a tuple of algorithms $\mathsf{TWS}=(\Gen, \sign, \verify)$ defined as follows:
    \begin{description}
        \item[$\Gen(1^\lambda) \to \pk{}$]:
            The Gen algorithm takes the predicate $P:\{0,1\}^* \times \{0, 1\}^* \to \{0,1\}$ and a security parameter $\lambda$ as inputs, and outputs a public key $\pk{}$. The input to the predicate is a public tag $t$ and a secret witness $w$.
        \item[$\sign(t, \pk, w, m) \to \sigma$]:
            The Sign algorithm takes as input a tag $t$, the public key $\pk{}$, a witness $w$ and a message $m$, and outputs a signature $\sigma$.
        \item[$\verify(t, \pk{}, m, \sigma) \to 0/1$]:
            The Verify algorithm takes a tag $t$, the public key $\pk{}$, a message $m$ and a signature $\sigma$ as inputs, and outputs a bit either $0$ (reject) or $1$ (accept).
    \end{description}
\end{definition}

\begin{definition}[Completeness]
A \TWS for a predicate $P:\{0,1\}^* \times \{0, 1\}^* \to \{0,1\}$, achieves completeness if for all tag and witnesses $t, w$ such that $P(t, w) = 1$ and message $m$, and sufficiently large security parameter $\lambda$, we have:
\[
\Pr\left[
\begin{matrix}
    \pk \gets \Gen(1^{\lambda}),
    \\
    \sigma \gets \sign(t, \pk, w, M):
    \\
    \verify(t, \pk, m, \sigma)=1
\end{matrix}
\right]\ge 1-\negl(\lambda) \; \cdot
\]
\end{definition}

\begin{definition}[Unforgeability]\label{def:UFCMA_TWS} Let $\mathsf{TWS}:=(\Gen,\sign,\verify)$ be a \TWS for a predicate $P$. The advantage of a PPT adversary $\advA$ playing the security game in Fig.~\ref{fig:tws_cma}, is defined as:

\[Adv^{\twsuf}_{\mathsf{TWS},\advA}(\lambda):=\Pr[\mathbf{Game}^{\twsuf}_{\advA,\mathsf{TWS}}(1^\lambda)=1]\]

A $\mathsf{TWS}$ achieves unforgeability against chosen tag and message attack if we have $Adv^{\twsuf}_{\mathsf{TWS},\advA}(\lambda)\le \negl(\lambda)$.

\begin{figure}[!htb]
\centering
\begin{small}
\begin{fboxenv}
\begin{varwidth}[t]{.5\textwidth}
    \begin{onehalfspace}
    \uline{$\mathbf{Game}^{\twsuf}_{\advA,\mathsf{TWS}}(1^\lambda)$:}	\\
        $\pk \gets \Gen(1^\lambda)$\\
         $Q_w \gets \emptyset,\ Q_s \gets \emptyset$
        \\
        $\left(t^*, m^*,\sigma^*\right) \sample \advA^{\oracle^{\sign}(\cdot),\oracle^{\Wit}(\cdot)}(\pk)$\\
        $\pcreturn \left( (t^*, m^*) \notin \mathcal{Q}_s ~ \land ~ t^* \notin \mathcal{Q}_w \land \verify(t^*, \pk,m^*,\sigma^*)\right)$
    \end{onehalfspace}
\end{varwidth}
\end{fboxenv}
\begin{fboxenv}
\begin{varwidth}[t]{.5\textwidth}
    \begin{onehalfspace}
        \uline{$\oracle^{\Wit}{(t)}$:}	\\
        Obtain $w$, s.t. $P(t, w) = 1$ \\
        $\mathcal{Q}_w \gets \mathcal{Q}_w \cup \{t\}$ \\
        $\pcreturn~ w$
    \end{onehalfspace}
\end{varwidth}
~~\vline\hspace{2.5pt}
\begin{varwidth}[t]{.5\textwidth}
    \begin{onehalfspace}
        \uline{$\oracle^{\sign}(t, m)$:}\\
        Obtain $w$, s.t. $P(t, w) = 1$ \\
        $\sigma \gets \sign\left(t,\pk,w,m\right)$ \\
        $\mathcal{Q}_s \gets \mathcal{Q}_s \cup \{(t, m)\}$ \\
        $\pcreturn~  \sigma$
    \end{onehalfspace}
\end{varwidth}
\end{fboxenv}
\end{small}
\caption{The Unforgeability Security Game.} 
\label{fig:tws_cma}
\end{figure}

\end{definition}

In simple terms, this definition addresses the situation in which the witnesses used to generate signatures might become known to an adversary. Specifically, we model witness leakage through a witness oracle: If it is computationally hard for the adversary to obtain a witness for a new tag, the property of unforgeability ensures that the adversary cannot create a signature for this fresh tag.

As noted before, \sysname requires modeling a notion of ``tag freshness''.
This can be done by simply setting time to be one of the components of a tag. 
Since an adversary can request witnesses corresponding to any tag of its choosing, it can request witnesses for old tags, thus modeling leakage of old witnesses.
Note how the unforgeability definition is agnostic to how a higher-level protocol defines what it means for a tag to be ``fresh''.

\begin{definition}[Witness Hiding]\label{def:witnesshiding_TWS}
    A \TWS for a predicate $P$, $\mathsf{TWS}:=(\Gen,\sign,\verify)$, achieves Witness-Hiding property if for all PPT adversaries $\advA$, there exist simulators $(\SimGen, \simsign)$, playing the described security games in Fig.~\ref{fig:tws_wh} and we have:
\[
    \left| 
     \begin{split}
        &\Pr[\textbf{Expt-Real}_{\advA}(1^{\lambda})=1]
        \\ &
        -  \Pr[\textbf{Expt-Sim}_{\advA}(1^{\lambda})=1]
     \end{split}
     \right| \le \negl(\lambda) \; .
\]
    
\end{definition}

\begin{figure}[!htb]
\centering
\begin{small}
\begin{fboxenv}
\begin{varwidth}[t]{.5\textwidth}
    \begin{onehalfspace}
        \uline{$\textbf{Expt-Real}_{\advA}(1^{\lambda}):$}	\\
        $\pk \gets \Gen(1^\lambda)$\\
        $b \sample \advA^{\oracle^{\sign}(\cdot, \cdot) 
        }(\pk)$ \\
        $\pcreturn~ b$ \\
        \uline{$\oracle^{\sign}(t, m)$:}\\
        Obtain $w$, s.t. $P(t, w) = 1$ \\
        $\sigma \gets \sign\left(t,\pk,w,m\right)$ \\
        $\pcreturn~  \sigma$
        \end{onehalfspace}
\end{varwidth}
~\vline~
\begin{varwidth}[t]{.5\textwidth}
    \begin{onehalfspace}
        \uline{$\textbf{Expt-Sim}_{\advA}(1^{\lambda}):$}	\\
        $(\pk, \td) \gets \SimGen(1^\lambda)$\\
        $b \sample \advA^{\oracle^{\simsign}(\cdot, \cdot)
        }(\pk)$ \\
        $\pcreturn~ b$ \\
        \uline{$\oracle^{\simsign}(t, m):$} \\
        $\sigma \gets \simsign\left(t,\pk, \td, m\right)$ \\
        $\pcreturn~  \sigma$
    \end{onehalfspace}
\end{varwidth}
\end{fboxenv}
\end{small}
\caption{The Witness Hiding Security Game. 
}
\label{fig:tws_wh}
\end{figure}

This characterizes the idea that an adversary gains no additional information about the witness associated with tags by observing the signatures. Essentially, this defines a privacy property for witnesses.

We will construct a specific \TWS, in~\cref{sec:zklogin-details}, as 
the core cryptographic component of \sysname. We also develop a more generic construction in\appref{sec:nizk-witsig}. 





\section{The \sysname system}
\label{sec:zklogindetails}

The main goal of \sysname is to allow users to maintain blockchain accounts leveraging their existing OpenID Provider accounts.

\subsection{Model}

There are four principal interacting \textit{entities} in \sysname:

\begin{enumerate}
    \item \textbf{OpenID Provider (OP)}: This refers to any provider that supports OpenID Connect, such as Google. A key aspect of these providers is their ability to issue a signed JWT containing a set of claims during the login process. For more details, see \cref{sec:oidc}. In our formalism below, we assume that each OP uses a fixed signing key pair for simplicity. We omit detailed formalism for handling key rotation. 
    \item \textbf{User}: End users who own the \sysname address and should have the capability to sign and monitor transactions. They are assumed to hold an account with the OP and may possess limited computational resources.
    \item \textbf{Application}: The application coordinates the user's authentication process. It comprises two components: the Front-End (FE), which can be an extension, a mobile or a web app, and optionally, a Back-End (BE). 
    \item \textbf{Blockchain}: The blockchain is composed of validators who execute transactions. In this work, we focus on public blockchains, e.g., Ethereum and Sui, where the entire state is public.\footnote{Extending \sysname to privacy-preserving blockchains, e.g., ZCash, and Aleo, is an interesting direction for future work.} \sysname requires support for on-chain ZKP verification and oracles to fetch the latest JWK (OP's public key), features commonly supported on many public blockchains. 
\end{enumerate}





\paragraph{Adversarial model} 
We assume that the app's backend is untrusted whereas its frontend is trusted.
This is reasonable because the frontend code of an 
app is typically public as it gets deployed on user's devices, and is thus subject to greater public scrutiny.

We assume that the OpenID Provider (OP) is trusted. This is reasonable because the main goal of our system is to design a user-friendly wallet. 
This does not make the OP a custodian since \sysname works with existing unmodified API and the OP is not even required to know about the existence of \sysname.

\subsubsection{Syntax}
We formally define \sysname to consist of the following algorithms:

\begin{definition}[\sysname]
    A \sysname scheme, is a tuple of algorithms $\mathsf{zkLogin}=(\Gen,$ $ \mathsf{zkLoginSign}, $ $\mathsf{GetWitness}, \mathsf{zkLoginVerify})$ defined as follows:
    \begin{description}
        \item[$\mathsf{zkLoginGen}(1^\lambda) \to \pk{}$]:
            The $\mathsf{zkLoginGen}$ algorithm takes the security parameter $\lambda$ as input, and outputs a public key $\pk{}$. 
        \item[$\mathsf{zkLoginSign}(\pk, \zkaddr, \iss, M, T_{exp}) \to \sigma$]: 
            The $\mathsf{zkLoginSign}$ algorithm takes as input an address $\zkaddr$, an issuer identifier $\iss$, a message $M$, an expiry time $T_{exp}$ and outputs a signature $\sigma$.
        \item[$\mathsf{GetWitness}(
        \iss, \zkaddr, T_{exp})$ $\to w
        $]:
            The $\mathsf{GetWitness}$ algorithm takes as input an issuer $\iss$, and address $\zkaddr$, and an expiry time $T_{exp}$, and outputs a witness $w$.
        \item[$\mathsf{zkLoginVerify}(\pk{}, \zkaddr, \iss, M, \sigma, T_{cur}) \to 0/1$]:
            The $\mathsf{zkLoginVerify}$ algorithm takes a public key $\pk{}$, an address $\zkaddr$, an issuer $\iss$, a message $M$, a signature $\sigma$, and a current time $T_{cur}$ as inputs, and outputs a bit either $0$ (reject) or $1$ (accept).
    \end{description}
\end{definition}

\subsubsection{Properties}
The completeness property of \sysname is described in\appref{app-prelims}.
We require \sysname to guarantee that unintended entities should not be able to perform certain actions or gain undesired visibility.


\paragraph{Security}
The main security property is a form of  unforgeability:
like in any secure signature scheme, an adversary should not be able to sign messages on behalf of the user.
In addition, we also want to prevent signatures on transactions based on expired JWTs. 

\begin{definition}[zkLogin Security]\label{def:zkLogin-sec} 
The advantage of a PPT adversary $\advA$ playing the security game in Fig.~\ref{fig:zklogin-sec}, is defined as:

\[Adv^{\mathsf{Sec}}_{\mathsf{zkLogin},\advA}(\lambda):=\Pr[\mathbf{Game}^{\mathsf{Sec}}_{\advA,\mathsf{zkLogin}}(1^\lambda)=1]\]

\sysname achieves security if we have $Adv^{\mathsf{Sec}}_{\mathsf{zkLogin},\advA}(\lambda)\le \negl(\lambda)$.

\end{definition}

\begin{figure}[!htb]
\centering
\begin{small}
\begin{fboxenv}
\begin{varwidth}[t]{.5\textwidth}
    \begin{onehalfspace}
    \uline{$\mathbf{Game}^{\mathsf{Sec}}_{\advA,\mathsf{zkLogin}}(1^\lambda)$:}	\\
        $\pk \gets \mathsf{zkLoginGen}(1^\lambda)$ 
        \\
        $Q_w \gets \emptyset,\ Q_s \gets \emptyset$
        \\
        $\left(\zkaddr^*, \iss^*, m^*,\sigma^*, T_{cur}^* \right) \sample \advA^{\oracle^{\mathsf{zkLoginSign}}(\cdot),\oracle^{\mathsf{GetWitness}}(\cdot)}(\pk)$\\
        Let $E$ be the event that $\forall (T_{exp}, \zkaddr^*, \iss^*) \in Q_w: T_{cur}^* > T_{exp}$ \\
        Let $F$ be the event that $\forall (T_{exp}, \zkaddr^*, \iss^*, m^*) \in Q_s: T_{cur}^*>T_{exp}$ \\
        $\pcreturn  ~ E \land F \land \mathsf{zkLoginVerify}(\pk, \zkaddr, \iss, m^*,\sigma^*, T_{cur}^*)$
    \end{onehalfspace}
\end{varwidth}
\end{fboxenv}
\begin{fboxenv}
\begin{varwidth}[t]{0.5\textwidth}
    \begin{onehalfspace}
        \uline{$\oracle^{\mathsf{GetWitness}}{(
        \iss, \zkaddr, T_{exp})}$:}	\\
        $w \gets \mathsf{GetWitness}(
        \iss, \zkaddr, T_{exp})$ \\
        $\mathcal{Q}_w \gets \mathcal{Q}_w \cup \{(T_{exp}, \zkaddr, \iss)\}$ \\
        $\pcreturn~ w$
        \\
        \uline{$\oracle^{\mathsf{zkLoginSign}}(\zkaddr, \iss, m, T_{exp})$:}\\
        $\sigma \gets \mathsf{zkLoginSign}(\pk, \zkaddr, \iss, m, T_{exp})$ \\
        $\mathcal{Q}_s \gets \mathcal{Q}_s \cup \{(T_{exp}, \zkaddr, \iss, m)\}$ \\
        $\pcreturn~  \sigma$
    \end{onehalfspace}
\end{varwidth}
\end{fboxenv}
\end{small}
\caption{The zkLogin Security Game.} 
\label{fig:zklogin-sec}
\end{figure}


\paragraph{Unlinkability}
This property captures the inability of any party (except the app) to link a user's off-chain and on-chain identities.
That is, no one can link a user's OP-issued identifier, the app they used, or any other sensitive field in the JWT, with their \sysname-derived blockchain account. 
The only exception is the $\iss$ claim, i.e., the unlinkability property does not mandate that the issuer be unlinkable.


We formalize this below.
At a high level, given 2 adversarially indicated claim sets  $\claimset_0$ and $\claimset_1$ (recall that a claim set is the list of claims present in a JWT), the adversary cannot link which one corresponds to a given \sysname address $\zkaddr$, even given access to several \sysname signatures for any address of its choosing.
If either of the claim sets $\claimset_0$ or $\claimset_1$ belong to a user controlled by the adversary, an adversary can win the game trivially~-- so both $\claimset_0$ and $\claimset_1$ must correspond to honest users.

Note that in some \sysname modes, we relax the unlinkability property from certain entities, e.g., the application or OP. 
In such instances, the adversary below can be any other party except the exempt entity.
\arnab{This paragraph is not clear - user-friendly how? Who are ``any other party"?} \deepak{How is explained later. Expanded what I meant by ``any other party"}

\begin{definition}[Unlinkability]\label{def:zkLogin-ul}
    \sysname
    achieves unlinkability property if for all PPT adversaries $\advA$
    playing the described security games in Fig.~\ref{fig:unlinkability}, we have:
\[
    \left| 
        \Pr[\textbf{Game}^{\mathsf{UL}}_{\advA, \mathsf{zkLogin}}(1^{\lambda})=1]
        -  1/2
     \right| \le \negl(\lambda) \; .
\]
    
\end{definition}
\begin{figure}[!htb]
\centering
\begin{small}
\begin{fboxenv}
\begin{varwidth}[t]{.5\textwidth}
    \begin{onehalfspace}
    \uline{$\mathbf{Game}^{\text{UL}}_{\advA,\mathsf{zkLogin}}(1^\lambda)$:}	\\
        $\pk \gets \mathsf{zkLoginGen}(1^\lambda)$,\ 
        Sample $b \sample \{0, 1\}$
        \\
        $(\claimset_0, \claimset_1, st) \gets \advA_1(\pk)$ \\
        If $(\claimset_0.\iss \neq \claimset_1.\iss)$, then $\pcreturn~ b$ \\
        Construct $\zkaddr$ from $\claimset_b$ \\
        $ b' \gets \advA_2^{\oracle^{\mathsf{zkLoginSign}}(\cdot)}(st, \zkaddr) $ \\
        $\pcreturn~ b' = b$
    \end{onehalfspace}
\end{varwidth}
\end{fboxenv}
\end{small}

\caption{The zkLogin Unlinkability Game.} 
\label{fig:unlinkability}
\end{figure}

\subsection{System details}
\label{sec:zklogin-details}

We begin by explaining how we derive addresses in \sysname and then explain how \sysname works.

Fix an OpenID Provider ($\iss$). Typically, each application needs to manually register with the provider. In this process, the app receives a unique audience identifier ($\aud$), which is included in all the JWTs (see \cref{lst:jwt-payload}) generated by the provider meant to be consumed by the app.

\paragraph{Address derivation} A simple way to define a user's blockchain address is by hashing the user's subject identifier ($\sub$), app's audience ($\aud$) and the OP's identifier ($\iss$). 

More generally, \sysname addresses can be generated from any identifier given by the OpenID Provider, as long as it is unique for each user (meaning no two users have the same identifier) and permanent (meaning the user can't change it). We call such an identifier a ``Stable Identifier'', denoted by $\stableid$.

A good example of a Stable Identifier is the Subject Identifier ($\sub$), which the OpenID Connect spec requires to be stable~\cite{openid-connect-core-subjectid}.

Besides the subject identifier, other identifiers like email addresses, usernames or phone numbers might also meet these criteria.
However, whether an identifier is considered stable can differ from one provider to another. For instance, some providers like Google don't allow changing email addresses, but others might.

\paragraph{The necessity of Salt} An important privacy concern arises whenever the stable identifier is sensitive, such as an email address or a username.
Note that the subject identifier is also sensitive if the provider uses public subject identifiers, meaning if a user logs into two different apps, the same $\sub$ value is returned.
To address privacy concerns and prevent the stable identifier from being easily linked to a user's blockchain address, we introduce a  ``salt''~-- a type of persistent randomness.

With this approach, a user's \sysname address is 
\begin{equation}\label{eq:addr}
    \zkaddr = H(\stableid, \aud, \iss, \salt)\; .
\end{equation}

In certain specific settings, adding a salt is relatively less critical or in fact undesirable.
One example is when the stable identifier is already private, like in the case of providers that support pairwise identifiers. 
However, it can still be useful as the salt offers unlinkability from the OpenID Provider.
Saltless accounts are desirable in contexts where revealing the link between off-chain and on-chain identities is beneficial (see discoverability in~\cref{sec:novel-features}).

For the purpose of the discussion below, we assume the incorporation of a salt hereafter. 
If a user wants to avoid setting a salt, it can be set to either zero (or) a publicly known value.
Note that it is impossible to enforce the use of a salt (at a blockchain protocol level) although it is our recommended choice, and we can have different users with and without salts.


\begin{figure}[t]
\centering
\begin{tikzpicture}[node distance=1.5cm, auto, font=\small, >=stealth, every path/.style={thick}, rect node/.style={draw, fill=gray!20, rectangle}]

\node[rect node] (appfe) {App FE};
\node[rect node, above left=1.5cm and 1cm of appfe] (oauth) {OP};
\node[rect node, below=of appfe] (authority) {Blockchain};
\node[rect node, left=2.5cm of appfe] (frontend) {User};

\draw[->] (appfe) -- node[midway, below, sloped] {A. Request JWT} (frontend);
\draw[->] (frontend) -- node[midway, above, sloped] {B. Login} (oauth);
\draw[->] (oauth) -- node[midway, above, sloped] {C. Return JWT} (appfe);

\draw[->] (appfe) edge [out=30,in=90, loop] node[above] {2. Get Salt} (salt);
\draw[->] (appfe) edge [loop right] node {3. Generate ZKP} (appfe);
\draw[->] (appfe) -- node[midway] {4. Submit transaction} (authority);

\coordinate (center) at (barycentric cs:appfe=1,frontend=1,oauth=1);
\node at (center) {1. Get JWT};

\end{tikzpicture}
    \caption{The \sysname System Overview. OP and FE stand for OpenID Provider and Front-End, respectively. Salt management and ZKP generation can be done either on the client side or delegated to a backend. A new ZKP needs to be generated only once per session.}
    \label{fig:overview}
\end{figure}

\cref{fig:overview} depicts the system's workflow including four parts, explained in follows. The first two parts \circled{1} Get JWT, and \circled{2} Get Salt, describe the protocol flows for implementing the $\oracle^{\Wit(\cdot)}$ oracle. The next two parts \circled{3} Compute ZKP, and \circled{4} Submit Transactions, informally describe how the $\Gen(\cdot)$, $\sign(\cdot)$, and $\verify(\cdot)$ functions are deployed.
The construction is formalized as a \TWS scheme, $\systws$, in~\cref{fig:pi-zklogin}, over the predicate $P_{\sysname}$.

\begin{figure}
\renewcommand{\tag}{\mathsf{tag}}
\begin{small}
\begin{description}
    \item[$P_{\sysname}\left(
        \begin{matrix}
            \tag = (\curOPpk, \iss, \zkaddr, T),\\ 
            w = (\jwt, \salt, \nonceRand, vk_u, sk_u) 
        \end{matrix}
    \right)$:]  
        \begin{minipage}[t]{\linewidth}
        $\zkaddr = H(\jwt.\stableid, \jwt.\aud, \jwt.\iss, \salt)$ and
        $\jwt.\iss = \iss$ and
        $\jwt.\nonce = H(vk_u, T, r)$ and
        $JWT.\verify(\curOPpk, \jwt)$ and 
        $(sk_u, vk_u)$ is a valid sig-key-pair.
        \end{minipage}
    \vspace{1em}
    \item[$\pzk\left(
            \begin{matrix}
                \zkx = (\curOPpk, \iss, \zkaddr, T, vk_u), \\
                \zkw = (\jwt, \salt, \nonceRand)
            \end{matrix}
        \right)$:]
        \begin{minipage}[t]{\linewidth}
        $\zkaddr = H(\jwt.\stableid, \jwt.\aud, \jwt.\iss, \salt)$ and 
        $\jwt.\iss = \iss$ and 
        $\jwt.\nonce = H(vk_u, T, r)$ and
        $JWT.\verify(\curOPpk, \jwt)$. 
        \end{minipage}
    \bigskip
    \item[\uline{$\Gen(1^\lambda)$}:] \
    \begin{itemize}
        \item Let $\Pi = (\setup, \prove, \verify)$  be a NIZK scheme.
        \item Sample $zkcrs \gets \Pi.\setup(1^\lambda, \pzk)$.
        \item Output $pk = zkcrs$.
    \end{itemize}
    \item[\uline{$\sign(\tag, \pk{}, w, M)$}:] \ 
        \begin{itemize}
            \item Parse $\tag$ as $(\curOPpk, \iss, \zkaddr, T)$. 
            \item Parse $\pk{}$ as $zkcrs$.
            \item Parse $w$ as $(\jwt, \salt, \nonceRand, vk_u, \sk{u})$.
            \item Set $\sigma_u \gets Sig.\sign(\sk{u}, M)$.
            \item Set $\zkx \gets (\curOPpk, \iss, \zkaddr, T, vk_u)$.
            \item Set $\zkw \gets (\jwt, \salt, r)$.
            \item Set $\pi \gets \Pi.\prove(zkcrs, \zkx, \zkw)$.
            \item Output $\sigma = (vk_u, T, \sigma_u, \pi)$.
        \end{itemize}
    \item[\uline{$\verify(\tag, \pk{}, M, \sigma)$}:] \ 
        \begin{itemize} 
            \item Parse $\tag$ as $(\curOPpk, \iss, \zkaddr, T)$. 
            \item Parse $\sigma$ as $(vk_u, T, \sigma_u, \pi)$.
            \item Set $\zkx \gets (\curOPpk, \iss, \zkaddr, T, vk_u)$.
            \item Verify $Sig.\verify(vk_u, \sigma_u, M)$.
            \item Verify $\Pi.\verify(zkcrs, \pi, \zkx)$.
        \end{itemize} 
\end{description}
\end{small}
\caption{\TWS, $\systws$. Address is derived from the stable identifier, e.g., $\stableid = \sub$.}
\label{fig:pi-zklogin}
\end{figure}

\begin{figure}
\renewcommand{\tag}{\mathsf{tag}}
\begin{small}
\begin{description}
    \item[\uline{$\Gen(1^\lambda)$}:] \
    \begin{itemize}
        \item Sample $pk \gets \systws.\Gen(1^\lambda)$.
        \item Output $pk$.
    \end{itemize}
    \item[\uline{$\mathsf{zkLoginSign}(\pk, \zkaddr, \iss, M, T_{exp})$}:] \
        \begin{itemize}
            \item Obtain $(w, \curOPpk) \gets \mathsf{GetWitness}(\iss, \zkaddr, T_{exp})$. 
            \item Let $\tag \gets (
            \curOPpk, 
            \iss, \zkaddr, T_{exp})$.
            \item Output $\sigma \gets \systws.\sign(\tag, \pk, w, M)$.
        \end{itemize}
    \item[\uline{$\mathsf{GetWitness}(\iss, \zkaddr, T_{exp})$}:] \
        \begin{itemize}
            \item Obtain $\curOPpk$ from JWK of $\iss$.
            \item Sample $(vk_u, \sk{u}) \gets Sig.Gen(1^\lambda)$.
            \item Sample $\nonceRand \gets \{0, 1\}^{\lambda}$.
            \item Set $\nonce \gets H(vk_u, T_{exp}, r)$.
            \item Obtain $\stableid, \aud, \salt$ from the User/App.
            \item Obtain $\jwt \gets \issueJWT(\curOPsk, \{\stableid, \aud, \iss, \nonce\})$ from the OP.
            \item Set $w \gets (\jwt, \salt, \nonceRand, vk_u, \sk{u})$.
            \item Output $(w, \curOPpk)$.
        \end{itemize} 
    \item[\uline{$\mathsf{zkLoginVerify}(\pk,\zkaddr,\iss, M, \sigma, T_{cur})$}:] \
        \begin{itemize}
            \item Output $0$, if $T_{cur} > \sigma.T$.
            \item Obtain $\curOPpk$ from JWK of $\iss$.
            \item Let $\tag \gets (\curOPpk, \iss, \zkaddr, \sigma.T)$.
            \item Output $\systws.\verify(\tag, \pk, M, \sigma)$.
        \end{itemize}
\end{description}
\end{small}
\caption{The signature scheme of \sysname using $\systws$. 
}
\label{fig:zklogin-system}
\end{figure}

\paragraph{\circled{1} Get JWT}
One of the key ideas in \sysname is to treat the OpenID Provider as a certificate authority by embedding data into the nonce during the OpenID flow~\cite{EPRINT:HMFGLMMPU23}.

The application generates an ephemeral key pair $(\ephpk, \ephsk)$, sets the key pair's expiry time $\expiryTime$, generates a randomness~$\nonceRand$ and computes the nonce via
\[\nonce \gets H(\ephpk, \expiryTime, \nonceRand)\; .\] 
Note that the expiration time, $\expiryTime$, must use a denomination that can be understood by the blockchain validators, e.g., ``$\expiryTime = \text{epoch } \#100$'' if the blockchain operates in epochs.
To prevent apps from setting arbitrarily long expiry times, blockchain validators can enforce constraints over its length, e.g., ensure that $\expiryTime < \curTime + \expiryDelta$ where $\curTime$ is the current epoch number and $\expiryDelta$ is the maximum number of epochs that a key pair can remain valid for (set by the blockchain).
The randomness~$\nonceRand$ helps achieve unlinkability as it prevents the OP from learning the ephemeral public key.

Next the app initiates an OAuth flow where the user logs in to the OP. 
This step may involve opening of a pop-up window asking for user's consent if it is the first time.
After the user successfully authenticates, the app receives a JWT from the OP, $\jwt \gets \issueJWT(\OPsk, \{\stableid, \aud, \iss, \nonce, \ldots\})$.
In essence, the JWT acts as a certificate over $\ephpk$, i.e., the JWT asserts that the owner of $\ephsk$ is indeed the same as the user identified by the OP-issued $\sub$.

\paragraph{\circled{2} Get Salt} 
As noted before, we recommend the use of an additional salt for unlinkability. 
Managing the salt, however, poses an operational challenge as losing the salt implies that the assets will be permanently locked. We present two approaches to salt management: either persist it on the client-side or in an app-managed salt service, discussed in depth in~\cref{sec:zklogin-modes}.
For this discussion, we assume that the salt is somehow fetched to the app's front-end.

\paragraph{\circled{3} Compute ZKP} The next step is to use the salt and the JWT to compute a Zero-Knowledge Proof proving the association between the ephemeral public key, $\ephpk$, and the address, $\zkaddr$. The ZKP's public inputs and witnesses are:

\begin{itemize}
    \item Public inputs: OP's public key, $\OPpk$, user's address, $\zkaddr$, the ephemeral public key, $\ephpk$, and its expiration time $\expiryTime$, i.e., $P = (\OPpk, \iss, \zkaddr, \ephpk, \expiryTime)$.
    \item Witnesses: $\jwt$, $\salt$ and the nonce, randomness $\nonceRand$.
\end{itemize}

The ZKP formally proves the predicate, $\pzk$, depicted in~\cref{fig:pi-zklogin}. The setup process, $\Pi.\setup$, for the NIZK system, $\Pi$, is run at the beginning to generate the $zkcrs$. This process only needs to be done once and the generated $zkcrs$ can used for all users and OPs. 
Formally, this is part of the $\Gen$ function of the \TWS, $\systws$. Informally, $\pzk$ captures the following steps:

\begin{enumerate}
    \item Hashing the claims $\stableid, \aud, \iss$ (extracted from the JWT) with the salt gives the expected address $\zkaddr$.
    \item Hashing the ephemeral public key, $\ephpk$, expiry time, $\expiryTime$, and the randomness~$\nonceRand$ gives the expected nonce (extracted from the JWT).
    \item The JWT verifies, i.e., $\verifyJWT(\OPpk, \jwt)$.
\end{enumerate}

The above Zero-Knowledge Proof needs to be generated on the client device for maximum privacy. 
We also offer an option to delegate computation of the ZKP securely, discussed in \cref{sec:zkp-generation}.







\paragraph{\circled{4} Submit transaction} 
Say that the transaction data is $tx$ and the Zero-Knowledge Proof generated in the previous step is~$\pi$. The app uses the ephemeral private key to sign the transaction, i.e., set $\sigma_u \gets Sig.\sign(\sk{u}, tx)$.
The final \sysname signature on the transaction $tx$ is $(\ephpk, \expiryTime, \sigma_u, \pi)$. 
Each validator can verify it by:
\begin{enumerate}
    \item Verifying the ZKP, $\pi$, with the public inputs $P = (\OPpk, \iss,$ $ \zkaddr, \expiryTime, \ephpk)$.
    \item Verify that $\OPpk$ is indeed the current public key of the OpenID Provider $\iss$. This step requires an oracle posting the public keys on-chain. For example, the oracles can access the JWK endpoint periodically (e.g., say every hour) and consider all JWKs seen in the last $\Delta$ epochs as current.
    \item Verify $Sig.\verify(\ephpk, \sigma_u, tx)$.
    \item Verify $\expiryTime \geq \curTime$ and $\expiryTime < \curTime + \expiryDelta$. 
\end{enumerate}

Finally, if all the conditions hold then the validators can execute the transaction~$tx$ sent by the address~$\zkaddr$.

The steps above of accessing the current time and the public key of the OP are part of the \sysname system that are not captured by the \TWS formalism. We assume that they are obtained correctly to ensure unforgeability and enforce freshness of the tag. \deepak{@Arnab: should we mention to the new system-level formalism here?}

\subsection{\sysname Modes}
\label{sec:zklogin-modes}

We now discuss two key practical considerations when using \sysname: managing salts and generating Zero-Knowledge Proofs.

\begin{table}[t]
    \centering
    \begin{tabular}{p{3.4cm}|p{4.6cm}}
        Mode & Considerations \\
        \toprule
         Salt service managed (uses JWT as one auth factor) & Salt service can link sensitive JWT fields (e.g., stable ID) to on-chain address (enclaves or MPC avoids this) \\
         \midrule
         User managed & Edge cases such as cross-device sync and device loss need to be handled \\
        \midrule
        \midrule
         Delegate ZKP generation & ZK service can link sensitive JWT fields (e.g., stable ID) to on-chain address (enclaves or MPC avoids this) \\
         \midrule
         Local ZKP generation &  Can be slow on resource-constrained devices
    \end{tabular}
    \caption{Privacy and Usability considerations for an app when using zkLogin. Security is unconditional in all modes.}
    \label{tab:tradeoffs}
\end{table}

\iffullpaper

\begin{table}[]
    \centering
    \begin{tabular}{c|c|c}
         Salt storage &
         Unlinkability & 
         JWT Privacy \\
         \toprule
         Plaintext & No & No \\
         \midrule
         Enclave & Yes & Yes \\
         \midrule
         Plain MPC & Yes & No \\
         \midrule
         \begin{tabular}{@{}c@{}}MPC with ZK and \\  secret-shared stable ID \end{tabular} & Yes & Yes
    \end{tabular}
    \caption{Privacy properties of different salt service instantiation choices (assuming JWT is an auth factor). Unlinkability and JWT Privacy refer to guarantees w.r.t the app.}
    \label{tab:salt-service-choices}
\end{table}

\fi

\subsubsection{Salt management}
\label{sec:salt-management}

An app can  manage their users' salts in two ways: run a salt service (either their own or run by a third party, e.g., a committee of nodes) or design flows that allow the user to manage salt themselves.

\paragraph{Salt service managed}
The basic idea is to employ a \textit{salt service} that stores users' salts and returns them upon proper authentication.
Any reasonable authentication policy can be used, with the most convenient being the submission of a valid JWT (others common authentication factors like TOTP, passkeys~\cite{lassak2024aren} or a combination are also possible).
We can deterministically derive the salt from the JWT's fields and a persistent secret seed $\saltSeed$ using a PRF, as follows: $\salt = F(\saltSeed, \sub \| \aud \| \iss)$.\footnote{We can add a counter to the salt derivation function to allow users to maintain multiple fully isolated on-chain accounts using the same OP and app.} 
This approach has the benefit that the salt service needs to maintain only a small secret that is independent of the number of users.
Note that the function $F$ can be any efficient PRF, e.g., a hash function or a deterministic signature scheme, e.g., BLS, EdDSA.

Note that since users have completely isolated zkLogin accounts with different apps, it may be acceptable for an app to learn the salts of its users for reasons beyond our protocol. 
However, in other cases, it maybe desirable to offer additional privacy.
So we consider two designs for the salt service:
\begin{enumerate}
    \item \textbf{Enclave}: Run an enclave, e.g., a TPM or TEE, that secures the salt seed. We discuss this approach further in \cref{sec:implementation}. 
    \item \textbf{Plain MPC}: Suppose the salt service is composed of $n$ nodes (e.g., a conglomerate of different apps) and we employ $t$-out-of-$n$ secret sharing to split $\saltSeed$. We set $F$ to a threshold signature scheme, e.g., BLS or Schnorr. Now the user can send their JWT to all the MPC nodes and obtain a share of their salt from each node.
\end{enumerate}
Both the enclave and the MPC options achieve unlinkability, as neither the enclave operator nor any fraction of less than $t$ nodes can link the user's off-chain identity with their on-chain address.


TEEs are known for their vulnerability against side-channel attacks~\cite{nilsson2020survey}.
A side-channel leakage can reveal the salt seed and the user's JWTs, thus breaking unlinkability.

The effect of a side-channel leakage can be minimized.
We can hide most claims in the user's JWT from the enclave (or the MPC nodes) by employing ZKPs. 
Users prove JWT validity using a ZKP while revealing just the necessary claims: $\sub$, $\aud$ and $\iss$. 
This hides sensitive claims in the JWT.

Compared to using a single TEE, the MPC option makes the task of an attacker wishing to learn the salt seed harder as they need to break into multiple nodes (each MPC node could additionally employ TEEs for defence-in-depth).

However, note that in the previously laid out design, breaking into a single MPC node suffices to learn the user's JWT.
We could go one step further to hide the JWT from the committee nodes using general-purpose MPC techniques~\cite{CCS:Keller20} as outlined below (however, it is more expensive for users compared to the prior MPC proposal as the user needs to generate $n$ ZKPs):


\begin{enumerate}
    \item \textbf{Setup:} Given $n$ nodes that secret-share $\saltSeed$ like before. Set $F$ to be a MPC-friendly PRF like MiMC~\cite{maram2021candid}.
    \item \textbf{Request:} The user computes $n$ ZKPs. ZKP $\pi_i$ takes the JWT as a private input and reveals secret-shares $\sub^i$, $\aud^i$, $\iss$, JWT header and the public key.
    \item \textbf{Response:} Each MPC node~$i$ verifies the ZKP $\pi_i$. Then all the MPC nodes jointly compute $\salt = F(\saltSeed, \sub \| \aud \| \iss)$. 
\end{enumerate}

\iffullpaper
\Cref{tab:salt-service-choices} summarizes the privacy properties achieved by different options.
\fi


\paragraph{User managed}
A second approach is to \textit{let users maintain their own salts}.
This reduces dependency on apps to manage users' salts. 
However, the main concern is the burden on users to remember yet another secret.

This burden can be minimized to a large extent in many settings.
Apps can store the salt on the local storage of users' devices so that users do not need to manually enter it before for every transaction.
Many modern devices are equipped with local enclaves, and even on older devices, persisting a salt locally (e.g., in a browser's local storage) can be reasonable since it is less sensitive than a password or a mnemonic.

Apps taking this route may need to implement additional flows to handle edge cases, e.g., device loss, multi-device / multi-browser support. 
Cross-device sync is easier if all the devices belong to a single provider, e.g., across different Apple devices~\cite{AppleCrossDeviceSync}. 
Newer authentication technologies like passkeys~\cite{lassak2024aren} can also help. 
For example, salt management can be piggybacked on passkeys by using the output of a deterministic passkey signature scheme (e.g., EDDSA) over a stable identifier (e.g., user's email) as the user's salt.

\Cref{tab:tradeoffs} briefly summarizes the trade-offs between the different choices for managing salts.


\subsubsection{Zero-Knowledge Proof (ZKP) generation}
\label{sec:zkp-generation}

Today, generating a ZKP can be time-consuming in resource-constrained environments (e.g., an old mobile phone).
(We do however caution that the space of general-purpose ZK proving systems is rapidly evolving with recent developments~\cite{venkitasubramaniam2023ligetron} specifically tackling the above problem.)
So we consider delegating ZKP generation to a backend service, called the \textit{ZK service}. 
The key challenge is to delegate in a way that security and unlinkability hold.
We present two ways of delegation that offer increasing amount of privacy to the user.

\paragraph{Normal delegation} A simple approach to ZKP delegation is to send all the witnesses to the ZK service.
The service can then compute and return a ZK proof. 
Crucially, note that even a malicious ZK service cannot break security as the ephemeral private key is not revealed to the service.
Thus, the delegation is secure.
However, the ZK service can break unlinkability as it learns both the user's JWT and salt, allowing it to compute the blockchain address (\cref{eq:addr}).

\paragraph{Full-private delegation} Another option proposed by recent works \cite{chiesa2023eos, garg2023zksaas} is to delegate ZK proving to a committee of nodes such that the entire witness (including the JWT) is hidden from a colluding minority. 
A more performant option might be to instantiate the ZK service inside an enclave.
Either approach can offer unlinkability from the ZK service.

\subsection{Security Analysis}

We prove that the proposed \TWS, $\systws$, achieves the unforgeability and unlinkability properties, with formal proofs given in the full version of our paper\appref{sec:zklogin-as-tws}. 

\begin{restatable}{theorem}{zkloginuf}
\label{thm:sigma-zklogin-unforge}
       Given that $\Pi$ satisfies knowledge-soundness, and JWT and Sig are EUF-CMA secure, and $H(\cdot)$ is a collision-resistant hash function, the \TWS, $\systws$, achieves unforgeability (Def.~\ref{def:UFCMA_TWS}).
\end{restatable}

\begin{restatable}{theorem}{zkloginpv}
\label{thm:sigma-zklogin-privacy}
    Given that $\Pi$ satisfies zero-knowledge, the \TWS, $\systws$, achieves witness hiding (Def.~\ref{def:witnesshiding_TWS}).
\end{restatable}




We next state the security and unlinkability properties of \sysname, based on the above properties of $\systws$, with formal proofs given in\appref{sec:zklogin-as-tws}.

\begin{restatable}{theorem}{zkloginsec}
\label{thm:sigma-zklogin-sec}
       Given that $\systws$ satisfies unforgeability (Def.~\ref{def:UFCMA_TWS}), $\sysname$ achieves security (Def.~\ref{def:zkLogin-sec}).
\end{restatable}

We prove this theorem 
as a direct reduction from the unforgeability propery of \TWS.

\begin{restatable}{theorem}{zkloginul}
\label{thm:zklogin-ul}
    Given that $\systws$ satisfies witness hiding (Def.~\ref{def:witnesshiding_TWS}) and that $\zkaddr$'s are computed as hiding commitment to the claimsets, $\sysname$, achieves unlinkability (Def.~\ref{def:zkLogin-ul}).
\end{restatable}

We prove this theorem 
based on the Witness Hiding property of \TWS and the fact that $\zkaddr$ is a hiding commitment to the claimsets.

\paragraph{Security of \sysname} The unforgeability of $\systws$ implies that an adversary cannot forge a signature for a given tag, even if it gets access to witnesses corresponding to other tags. In \sysname, the system layer ensures that tags time out on a defined cadence, and hence witnesses need to be refreshed for new tags. Hence older tags are no longer useful for creating signatures. In particular, this means that a new JWT is needed for a \sysname signature, once the ephemeral public key expires. 

We now discuss security considerations by component.

\begin{itemize}
    \item \textit{Application:} There is a reliance on the app for liveness (an app that stops functioning leads to locked assets) but not for security (a malicious app cannot steal user's assets). 
    Users can hedge against this risk by setting up a Multi-sig between two or more apps, e.g., a 1-out-of-2 Multi-sig between zkLogin on app 1 and 2. 
    \item \textit{Salt service:} Assuming a salt service was employed, the security of \sysname still holds even if it acts maliciously.
    This is because the salt service does not have access to the ephemeral private key needed to sign transactions.
    \item \textit{OpenID Provider:} Whether a malicious OP (or equivalently, hacked OP accounts) can break unforgeability depends on the salt management strategy. 
    If a JWT alone is enough to fetch the salt from a salt service, then a malicious OP can sign arbitrary JWTs to break security. 
    Using a different salt management strategy avoids this, e.g., if the salt is user-managed or employ a second factor such as TOTP codes in salt service. Therefore, apps that desire a higher level of security must choose a salt management strategy appropriately.
\end{itemize}

\paragraph{Unlinkability of \sysname} 
\sysname achieves unlinkability based on the witness-hiding property of $\systws$ and on $\zkaddr$ being a hiding commitment to $\stableid$ and $\aud$. 

We now discuss privacy considerations by component.


\begin{itemize}
    \item \textit{Blockchain:} The blockchain records the transactions along with the \sysname signatures.
    These records are publicly visible and contain information about the \sysname address that signed it. However, the only JWT claim visible in a transaction is the OP identifier ($\iss$).
    \item \textit{OpenID Provider:} Like with unforgeability, if a JWT alone is enough to fetch a salt, then unlinkability is lost. Otherwise, \sysname achieves unlinkability against the OP. Consequently, the OP cannot track a user's transactions (without resorting to imperfect timing side-channels\footnote{JWTs can be pre-fetched and ZKPs can be reused for long periods of time masking any obvious correlations.}). However, a small privacy leak exists stemming from the way OAuth works: the OP knows the set of all users using a given app. 
    \item \textit{Salt service:} The salt service (if employed), by design, maintains users' salts.
    However, the use of an enclave (resp., MPC) hides the salt from a malicious enclave operator (resp., a fraction of MPC nodes). 
\end{itemize}

\iffullpaper
\subsection{Extension: Nonce-less OpenID Providers}



Even though the OpenID Connect spec requires providers to include a nonce when the request contains it, some providers do not. 
We now present a protocol that adapts \sysname for nonce-less providers. 
We remark that outside the realm of OpenID, several prominent identity documents, e.g., e-Passports~\cite{icao2010mrtd}, already contain a digital signature over user's biographic information~\cite{rosenberg2023zk}, and can thus be used with the below protocol.

The main idea is to bind the ephemeral public key and the expiration time by directly hashing it with the JWT, in line with the generic construction in\appref{sec:nizk-witsig}. The final signature is:
\begin{enumerate}
    \item Hash of the ephemeral public key, timestamp, and the JWT
    \item A ZKP, which proves: 
        \begin{enumerate}
            \item Consistency of the JWT claims with the address 
            \item Validity of the OP signature 
            \item Consistency of the hash from above (1)
        \end{enumerate}
\end{enumerate}
As we show in\appref{sec:nizk-witsig}, this construction achieves unforgeability and witness hiding properties as a \TWS.

However, this construction falls short of the security guarantees of our nonce-based construction in certain scenarios. In contrast to the standard \sysname construction, there is no nonce here to commit to the ephemeral public key and the expiration time inside the JWT. Hence anybody having access to a valid JWT can construct a valid signature, without needing to authenticate to the OP. Freshness can still be enforced by additionally checking the JWT's internal timestamps inside the ZKP, but there is less flexibility compared to our previous construction.   

Note that a malicious ZK service can sign transactions, as it gets access to fresh JWTs.
%
So the above protocol is only secure if the proofs are generated locally on user devices.
We leave how to make it more secure and user-friendly for future work.
\fi
\subsection{\sysname Novel Features}
\label{sec:novel-features}

Apart from easy onboarding, \sysname offers a few novel features that were not seen before to the best of our knowledge.

\paragraph{Discoverability}
This means that a user's existing digital identifier with which they are prominently identified (email, username, etc.) can be, for this feature, publicly bound to their blockchain address. 
While this obviously breaks unlinkability, this feature can be extremely useful in certain contexts where users want to maintain public profiles. 
For example, content creators may want to establish provenance by digitally signing their content~\cite{techreview2023cryptography}.
Users with an existing \sysname account can make their account discoverable by simply revealing the stable ID, audience ID and salt.
New users can create a discoverable account by avoiding the use of a salt.

\paragraph{Partial reveal}
It can also be useful to make an existing \sysname account partially discoverable, e.g., revealing only the audience ID or just a portion of their stable ID. 
The latter would allow employees of an organization to reveal that they belong to an organization without revealing their identity, e.g., if email is the stable ID, Alice (``alice@nyu.edu'') can reveal her university affiliation by revealing the email's domain name, i.e., ``@nyu.edu''.
Revealing just the TLD of an email can also have interesting applications, e.g., reveal you are a student (``.edu'') or belong to a particular country (``.uk'' implies a UK based email).

\iffullpaper

\paragraph{Anonymous blockchain accounts} \sysname allows the creation of a blockchain account that hides the identity of the account owner within an anonymity set.
Two different approaches are possible.
If the stable identifier has some structure, e.g., with emails, we can derive a blockchain address from a portion of the email address (either the domain name or the TLD), $\zkaddr = H(P, \aud, \iss)$ where $P$ is the relevant portion of the email (no salt). 

A different approach is needed for other identifiers that do not have such a structure or if greater control is needed over the anonymity set.
First decide the anonymity set, e.g., a list of Google Subject Identifiers~$L$.
The \sysname address is derived from the entire anonymity set, $\zkaddr = H(L, \aud, \iss)$. 
To sign a transaction, the user needs to prove that their Google ID belongs to the list~$L$.

\paragraph{Claimability}
This implies the ability to send assets to a person even before they have a blockchain account.
A sender can derive the receiver's \sysname address using the receiver's email (or a similar identifier) and randomly chosen salt. The salt can either be sent to the sender over a personal channel or simply set to a default value, e.g., zero. 
The former approach places a burden on the receiver to manage the salt whereas the latter makes the address discoverable.

\fi

\section{Implementation and Deployment}
\label{sec:implementation}

We have implemented \sysname in a production environment.

The key choice for practitioners is that of the proving system. 
Among the many available ones~\cite{EPRINT:BBHR18,EPRINT:GabWilCio19,EC:CHMMVW20}, we chose Groth16~\cite{EC:Groth16}  due to its mature tooling ecosystem and compact proof sizes.
We leverage the circom DSL~\cite{circom} to efficiently write up the R1CS circuit. 
\Cref{sec:circuit-details} provides an in-depth look at the circuit covering various optimizations that have helped reduce the number of R1CS constraints by at least an order of magnitude for some components.

Groth16 necessitates a circuit-specific trusted setup, and to this end, we have orchestrated a ceremony with the participation of over 100 external contributors to generate the Common Reference String (CRS). More details about our experience conducting the ceremony are in\appref{sec:ceremony}.


\Cref{sec:performance} evaluates \sysname's performance. 
The two main components of \sysname are ZK proof generation and salt management. 
As noted in~\cref{sec:zklogin-modes}, each application can choose to implement these components in different ways. 
In this section, we implement and benchmark one of the configurations.

\paragraph{ZK proof generation}
We focus on the na\"ive delegation approach where the entire witness is sent to the ZK service.

\paragraph{Salt management}
We implement and benchmark TEE-based salt management. 
As noted before, apps can manage salts in other ways, like user-managed (shared between owned devices) or MPC. 

Specifically, we employ AWS Nitro enclaves for this benchmark, although a similar service can be spun up using alternatives like Intel SGX or Google Cloud's TEE. We generate the salt seed, $\saltSeed$, inside a Nitro enclave, which is used only once for bootstrapping purposes. This seed is managed by the AWS Key Management Service (KMS) such that the seed is only decryptable inside enclaves whose measurements\footnote{An enclave's measurements include a series of hashes and platform configuration registers (PCRs) that are unique to the enclave.} match those of the salt service.\footnote{\href{https://docs.aws.amazon.com/enclaves/latest/user/kms.html}{https://docs.aws.amazon.com/enclaves/latest/user/kms.html}} To avoid sole dependency on AWS, the salt service operator may also want to back up the salt seed with a committee of reputable parties using Shamir secret-sharing.

The salt service enclave functions as described in~\cref{sec:zklogin-modes}.
In brief, it authenticates users using the JWTs and derives a unique salt from the salt seed.
Note that the enclave needs to be able to handle HTTPS requests internally to (a) fetch JWKs from a public HTTPS endpoint to verify the JWT and (b) decrypt incoming requests and encrypt outgoing responses from inside the enclave. 
Note that the salt service can either use HTTPS or rely on encryption to service user requests so that any servers-in-the-middle do not see the salt.
Additionally, common-sense practices, such as deleting user requests as soon as they are serviced, can help minimize any adverse effects of potential TEE side channels.

\subsection{ZKP implementation details}
\label{sec:circuit-details}

Despite offloading proof generation to an untrusted server, optimizing our R1CS circuit remains crucial for several reasons.
First, a larger circuit translates to increased operational complexity during the ceremony, requiring the transmission of significantly larger files.
Second, a more complex circuit incurs greater proving costs, both in terms of the time required to generate a proof and the expense of maintaining powerful servers.



Recall that the circuit takes a JWT~$J$ as input and parses certain claims, e.g., ``sub'' from it.
Accordingly, it has two main components: (i) validating the JWT, and (ii) parsing the JWT.
\ifvone
In total, our R1CS circuit has around 1.1 million (slightly above $2^{20}$) constraints.
\fi
Notably, the circuit utilizes the Poseidon hash function~\cite{grassi2021poseidon} for hashing in the circuit where possible.

\subsubsection{JWT validation}

One set of circuit operations is to verify the JWT signature. 
The IETF spec recommends the use of two algorithms for signing JWTs: RS256 and ES256~\cite{rfc7518-section3}.
Of the two, we found that RS256 is the most widely used, hence we only support that currently.  
RS256 is short for RSASSA-PKCS1-v1\_5 using SHA-256.
Verifying a message under RS256 involves two steps: 
\begin{enumerate}
    \item hash the JWT's header and payload with SHA-2 to obtain a hash $h$ and
    \item perform a modular exponentiation over a 2048-bit modulus: if the signature is $\sigma$ and modulus is $p$ (both 2048-bit integers), check if $\sigma^{65537} \% p = pad(h)$ where $pad()$ is the PKCS1-v1\_5 padding function.
\end{enumerate}

We have largely reused existing code for RS256 signature verification.
SHA-2 is the most expensive operation in the overall circuit taking up around 750k \ifvone (66\%) \fi of the constraints.
Big integer operations needed to perform modular exponentiation are the second most expensive operation taking around 155k \ifvone (14\%) \fi of the constraints.
For bigint operations, we leverage existing code\footnote{We used code from \href{https://github.com/doubleblind-xyz/double-blind}{https://github.com/doubleblind-xyz/double-blind}, \href{https://github.com/zkp-application/circom-rsa-verify}{https://github.com/zkp-application/circom-rsa-verify} for the RSA circuit. Other helper functions were inspired from \href{https://github.com/TheFrozenFire/snark-jwt-verify}{https://github.com/TheFrozenFire/snark-jwt-verify}.} which in turn implement efficient modular multiplication techniques from~\cite{xjsnark2018}.

\subsubsection{JWT parsing}

Parsing the JWT for relevant claims takes approximately \ifvone 235k constraints (20\%) \fi and is where most of our optimization efforts lie. 

\ifvone

A na\"ive approach to JWT parsing involves completely decoding the JWT and parsing the complete header and payload JSONs to extract all the relevant claims, e.g., ``sub'', ``aud'', ``iss'' and ``nonce'' from payload and ``kid'' from the header.
Fully parsing the JSON inside a ZK circuit involves encoding every rule in the JSON grammar, which is likely to be very inefficient. 


%


We address these challenges by selectively parsing and decoding relevant parts of the JWT, as explained below.

\begin{itemize}
    \item \textit{Public Input Header}: Instead of Base64 decoding the JWT header in the circuit, we reveal it as a public input. This public header is then parsed and validated as part of the \sysname signature verification process. Note that the header (cf.~\cref{fig:jwt}) does not contain sensitive or linkable claims.
    \item \textit{Selective Payload Parsing}: Next, we completely decode the JWT payload and selectively parse relevant portions of the resultant JSON. This is possible due to the following observations about JWT payloads: 
\begin{enumerate}
    \item Provider follows the JSON spec. In particular, it only returns valid JSONs and properly escapes all user-input strings in the JSON. 
    \item All the claims of interest are in the top-level JSON and the JSON values are either strings or boolean.
    \item Escaped quotes do not appear inside a JSON key (cf.~\appref{lst:quote-json-key}).
\end{enumerate}

\end{itemize}

Decoding a single Base64 character takes $73$ constraints, so decoding the entire payload of maximum possible length $L_{max} = 1500$ bytes (we set $L_{max}$ to 1500 based on empirical data) incurs $73 L_{max} \approx 110k$ constraints. 
Note that using lookup arguments~\cite{cryptoeprint:2022/1530} can reduce the Base64 decoding costs significantly. 

We then slice the portion of the JSON Payload containing a JSON key-value pair (a claim name and value) together with the ensuing delimiter, i.e., either a comma ``,'' or a right brace ``\}''.
It is important to include the delimiter as it indicates the end of the value.
In more detail, the JSON key-value pair parsing component of the circuit takes a (unparsed JSON) string~$S$ and does the following for every claim to be parsed:

\fi

\begin{lstlisting}[language=json, caption=Decoded JWT Payload., label=lst:simple-json]
{"sub":"123","aud":"mywallet","nonce":"ajshda"}
\end{lstlisting}

\begin{enumerate}
    \item Given a start index $i$ and length $l$, use string slicing techniques (see below) to extract the substring $S' = S[i: i + l]$. For example, if $S$ is as shown in~\cref{lst:simple-json}, $i = 1$ and $l = 12$, then $S' = $\texttt{\color{red} "sub":"123",}.
    \item Check that the last character of  $S'$ is either a comma ``,'' or a close brace ``\}''.
    \item Given a colon index $j$, check that  $S'[j]$ is a colon ``:''.
    \item Output $key = S'[0:j] = $ \texttt{\color{red} "sub"} and $value =  S'[j + 1: -1] = $ \texttt{\color{red} "123"}. In addition, while we do not explain here, our circuit can also tolerate some JSON whitespaces in the string~$S$ using similar techniques.
    \item Check that the key and value are JSON strings by ensuring that the first and last characters of $key$ (resp., $value$) are the start and end quote respectively.
\end{enumerate}

The above strategy is used to parse four claims in the circuit, namely, the stable ID (e.g., $\sub$ or $\sfemail$), nonce, audience and 
email verified (If the stable ID is email, we have to additionally check that the ``$\mathsf{email\_verified}$'' claim is true in order to only accept verified emails~\cite{openid-connect-core}).
Following this, the extracted claim values are processed, e.g., the stable ID is fed into the address derivation, the nonce value is checked against a hash of the ephemeral public key, its expiry and randomness, etc.

Note that it is possible for an attacker to over-extend the end index. Continuing the above example, an attacker could set $S' = $ \texttt{\color{red} "sub":"1320606","aud":"mywallet",}.
But this would have the effect of obtaining the JSON value of \texttt{ \color{red} "1320606","aud":"mywallet"}.
And, crucially, this is not a valid JSON string as per the JSON grammar because a JSON string would escape the double quotes.
This implies that no honest user will have this subject identifier.
Therefore, the above attempt does not lead to a security break.
Similarly, certain portions of the JWT allow user-chosen input, e.g., in nonce~-- but this does not lead to a security threat as you cannot inject unescaped key-value pairs into a JSON string.

The main security threat is of an attacker who inputs maliciously crafted inputs during the OpenID sign-in flow to sign transactions on behalf of a honest user (without having to steal the user's credentials).
We have argued that this is not possible under reasonable assumptions.

\paragraph{Slicing arrays}
Given an input array $S$ of length $n$, index $i$ and length $m$, we need to compute the subarray $S[i: i+m]$.

We start with a na\"ive slicing algorithm.
For each output index $j$, compute a dot product between $S$ and the $n$-length vector $O$ s.t. $O[j] = 1$.
Since the dot product involves $n$ multiplications, this takes $n$ constraints per output index, and a total of $n * m$ constraints.
Concretely, the value of $n$ is $L_{max} = 1600$ and the value of $m$ ranges between 50 and 200 (depending on the claim value length); so if $m=100$, slicing once incurs 160k constraints.
As we slice once for each JSON claim parsed (5 times in total), this is costly.

We observe that the default input width of elements in the JWT is only 8 bits, which is much smaller than the allowed width of a field element in BN254 (253 bits).
So we pack 16 elements together (a packed element is 128 bits), apply the na\"ive slicing algorithm over the packed elements, and finally unpack back to the original 8-bit width, while taking care of boundary conditions.
With this optimization, slicing an array costs about $18m + (n*m)/32$ constraints.
Using the above values of $n=1600$ and $m=100$, the number of constraints is only $33k$, i.e., a $4.8$x reduction per slice operation.
Overall, this trick reduces the number of constraints for slicing by more than an order of magnitude.





\subsection{Evaluation}
\label{sec:performance}

\begin{table}[]
\ra{1.2}
    \caption{Latency comparison between \sysname and Ed25519 signatures. The last row shows the time taken for signing a transaction and getting it confirmed on a test network. For \sysname, this includes the time taken to fetch salt and ZKP.}
    \centering
    \begin{tabular}{c|cc}
        Operation & \sysname & Ed25519  \\
        \toprule
        Fetch salt from salt service & 0.2 s & NA \\
        Fetch ZKP from ZK service & 2.78 s & NA \\
        Signature verification & 2.04 ms & 56.3 $\mu$s  \\ 
        \midrule
        E2E transaction confirmation & 3.52 s & 120.74 ms \\ 
        \bottomrule
    \end{tabular}

    \label{tab:latency}
\end{table}

We now evaluate the end-to-end performance of \sysname, along with micro-benchmarks for a TEE-based salt service and delegated ZK proof generation. 
We also discuss the impact of using \sysname (vs) traditional signatures on users and validators. 
\Cref{tab:latency} summarizes the important latency numbers.

\paragraph{TEE-based salt management} We have deployed the salt service within an AWS Nitro enclave running on an 
m5.xlarge\footnote{\href{https://aws.amazon.com/ec2/instance-types/m5/}{https://aws.amazon.com/ec2/instance-types/m5/}} instance, which boasts 4 vCPUs and 16 GB RAM.
The average response time for retrieving salt is 0.2 s.

\paragraph{Delegated ZKP generation} The ZK service is built around rapidsnark~\cite{rapidsnark}, a C++ and Assembly-based Groth16 prover. 
We have bechmarked ZKP generation on a Google cloud n2d-standard-16\footnote{\href{https://cloud.google.com/compute/docs/general-purpose-machines}{https://cloud.google.com/compute/docs/general-purpose-machines}} instance, which boasts 16 vCPUs and 64 GB RAM.

The peak RAM usage of the ZK service during the last three months of our production deployment is 1.19 GB with an average memory usage of 0.82 GB.
The heavy reliance on memory led us to run the ZK service in a way that each machine only handles one request at a time.
So we rely purely on horizontal scalability, i.e., adding more machines, to handle multiple requests simultaneously.

We now present the time consumed by various components of the ZK service in servicing a single request (assuming no contention).
Before calling rapidsnark, the service converts user inputs (e.g., JWT, salt) to a witness using a combination of TypeScript code and the circom-generated witness calculator~\cite{circom}.
The circom witness calculator takes 550.05 $\pm$ 22.42 ms (mean and standard deviation).
Next, the witness is used to generate a Groth16 proof using rapidsnark.
The average proof generation time is 2.1 $\pm$ 0.15 s.
In total, the end-to-end proof generation time is 2.78 $\pm$ 0.25 s. 
We report results averaged over 300 runs.
The final number includes time taken for other tasks like generating inputs to the witness calculator (TypeScript code) and DoS prevention measures such as JWT verification, proof caching that were not counted before.

The machine specs we've used are comparable to that of a high-end laptop or desktop. So our results suggest that it may already be practical to instantiate the ZK service locally on certain user devices.

\paragraph{\sysname end-user costs} The end-to-end latency experienced by an end-user when submitting a \sysname transaction is  3.52 $\pm$ 0.36 s, which includes the time taken for fetching the ZKP and salt from the ZK and salt services respectively, signing the transaction and getting it confirmed on a test network (not including the time taken for any UX pop-ups, e.g., interaction with the Google sign-in page).
On the other hand, the confirmation latency for a traditional signature is only 120.74 $\pm$ 5.32 ms. These numbers are averaged over 50 runs. 



Even though there is a notable latency increase, the end users do not perceive it in most cases due to a couple reasons.
A single ZKP can be reused across all the transactions signed in the same session. So the above number reports the latency for the first transaction of a session. Subsequent \sysname transactions incur a similar latency as that of a traditional signature.
Finally, even for the first transaction, applications can hide the extra few seconds by pre-fetching the Zero-Knowledge Proof in the background much before the user signs a transaction.



\paragraph{\sysname validator costs}
The verification of a \sysname signature takes 2.04 ms on an Apple M1 Pro with 8 cores and 16 GB RAM. This is about two orders of magnitude slower than verifying a EdDSA signature. 
The size of a \sysname signature is around 1300 bytes (when encoded in Base64), which is about an order of magnitude larger than EdDSA signatures.

Given the longer verification time and bigger signatures, we conduct a small stress test to understand the effect of verifying zkLogin signatures on blockchain validators.
We use a testbed of 8 validator nodes each with 8 cores, 128 GB RAM that are split across New York and Los Angeles.
We subject the testbed to a load of 1000 transactions per second (TPS) and observe how the throughput changes when we switch the signature scheme from EdDSA to zkLogin.
The testbed was able to process 850 TPS when using EdDSA-signed transactions as opposed to 750 TPS with zkLogin-signed transactions, i.e., roughly a 11\% decrease. 
We suspect that the drop is not big because signature verification is only a small part of what validators do. 


\section{Related Works}\label{sec:relatedworks}

Next, we classify the related works into a few categories.

\paragraph{Deployed OAuth wallets} Many wallet solutions leverage OAuth to onboard users onto blockchains. 
We are primarily interested in non-custodial OAuth-based wallets.
Prior approaches can be classified as:

    

   


\begin{itemize}


    \item (TEE) Use OAuth to authenticate to HSMs/Enclaves to verify the user's authentication token and retrieve secrets/attestation that can be used to sign transactions. E.g., Magic \cite{magic}, DAuth \cite{dauth}, Face Wallet~\cite{faceWallet}.

    \item (MPC) Use OAuth to authenticate to a non-colluding set of servers that either directly sign the transaction on the user's behalf, e.g., Web3Auth \cite{web3auth} and Near~\cite{nearoauth} uses threshold crypto for signing (or) retrieve secrets that are later used on the client side to sign transactions, e.g., Privy \cite{privy} only stores 1 out of 3 shares on the server. 
\end{itemize}


\sysname may also be viewed as a kind of 2PC between the OpenID provider and the app.
However, the main difference between \sysname and the above MPC wallets is that (a) \sysname relies on the app purely for liveness whereas MPC wallets rely on the app for security (to a varying degree depending on how many shares are held by the app), and (b)  its novel features like discoverability, partial reveal and claimability that were not possible before.

\paragraph{Other prior works} 
The approach of embedding arbitrary data into the OpenID Connect's nonce draws inspiration from recent works~\cite{EPRINT:HMFGLMMPU23, OPENZEPPELIN}.
However, a problem with these works is that they require showing the sensitive ID token to the verifier.
This is a big issue if the verifier is a public blockchain.

ZK Address Abstraction~\cite{EPRINT:PLLCCKCM23} is the closest prior work to ours: like us, they also use a general-purpose ZKP over the JWT to authenticate blockchain transactions. However, their approach has a few drawbacks which makes it impractical: 
(a) they tie the transaction closely into the ZKP, so the user would have to generate a new proof for every transaction,
(b) they assume providers use a ZK-friendly signature scheme EDDSA that is not currently used by any popular OpenID providers and 
(c) they force the user to generate the ZKP on their own.
Instead, we leverage the nonce embedding trick to reuse a single proof across many transactions and provide a choice to offload proof generation in a trustless manner.

CanDID~\cite{candid2021} enables the migration of web2 credentials to a blockchain without requiring modifications to existing providers, similar to \sysname.
However, CanDID introduces an additional MPC committee, a dependency that \sysname deliberately avoids. 
On the other hand, CanDID can port arbitrary credentials like SAML, username-password whereas \sysname only works with those that include a signature like OpenID Connect.


\paragraph{Subsequent works}
A few subsequent works~\cite{aptoskeyless,bonsaipay} also use OpenID Connect with ZKPs for onboarding users onto blockchains.
Bonsai Pay~\cite{bonsaipay}, for instance, is built on the Rust-based risc0 zkVM~\cite{risc0}, which facilitates a more accessible ZK development experience using a conventional programming language. However, its proof generation times are significantly slower, taking minutes to complete. Additionally, Bonsai Pay exposes the stable identifier (email) on-chain, a privacy concern that \sysname avoids through the use of a salt.

Aptos Keyless~\cite{aptoskeyless}, like us, uses a salt to enhance privacy but introduces notable centralization points in its design. In one mode, if the prover service becomes unavailable, no party can sign transactions, creating a critical single point of failure. Additionally, Aptos Keyless grants special powers to designated recovery applications, allowing users to access all their assets through these recovery apps, even if the original app is offline. However, this approach carries a serious risk: a vulnerability in one recovery app could lead to the loss of assets stored in \textit{all} the user accounts.





\section{Conclusion}
\label{sec:conclusion}

We have introduced \sysname, a novel approach for authenticating blockchain users by leveraging the widely-adopted OpenID Connect authentication framework.
Crucially, the security of a \sysname-based blockchain account relies solely on the security of the OpenID provider's authentication mechanism, avoiding the need for additional trusted third parties.
\sysname can be used on any public blockchain that supports Groth16 verification and has oracles to fetch the latest JWK, e.g., Ethereum, and Sui. 


At the heart of \sysname is a mechanism that utilizes a (signed) OpenID token to authorize arbitrary messages.
We formalized \sysname as a Tagged Witness Signature, an extension of Signatures of Knowledge capturing the leakage of old tokens.
We have employed Zero-Knowledge Proofs to conceal all sensitive details in an OpenID token.
Additionally, the inclusion of a salt effectively obscures any connection between an individual's off-chain and on-chain accounts.
We have validated \sysname's real-world viability with a fully functional implementation that is also currently deployed in a live production environment.


\bibliographystyle{ACM-Reference-Format}
\balance
\bibliography{localbib,cryptobib/abbrev3,cryptobib/crypto}

\iffullpaper
    \appendix
    \crefalias{section}{appendix}
    \section{Deferred Definitions}
\label{app-prelims}

\paragraph{Digital Signature Schemes} A signature scheme $\mathsf{S} = (\Gen, \sign, \verify)$ with message space $\mathsf{M}$ consists of three algorithms, defined as follows:

\begin{description}
  \item[$\Gen(1^\lambda) \to (\sk{},\pk{})$]: is a randomized algorithm that on input the security parameter $\lambda$, returns a pair of keys $(\sk, \pk)$, where $\sk{}$ is the signing key and $\pk$ is the verification key.
  \item[$\sign(\sk{}, m) \to \sigma$] takes as input the signing key $\sk{}$, and a message $m$, and returns a signature $\sigma$.
  \item[$\verify(\pk, m, \sigma) \to 0/1$] is a deterministic algorithm that takes as input the verification key $\pk$, a message $m$, and a signature $\sigma$, and outputs $1$ (accepts) if verification succeeds, and $0$ (rejects) otherwise.
\end{description}

A signature scheme satisfies correctness if for all $\lambda$, $m \in \mathsf{M}$, and every signing-verification key pair $(\sk, \pk) \gets \Gen(1^{\lambda})$, we have, $\verify(\pk, m, \sign(\sk, m)) = 1.$

\begin{definition}[EUF-CMA]
  \label{def:sign}
 Let $\mathsf{S} = (\Gen, \sign, \verify)$ be a signature scheme. The advantage of a PPT adversary $\advA$ playing the security game described 
in Fig.~\ref{fig:euf_cma}, is defined as:

\[Adv^{\mathsf{EUF-CMA}}_{\mathsf{S},\advA}(\lambda):=\Pr[\mathbf{Game}^{\text{EUF-CMA}}_{\advA,\mathsf{TWS}}(1^\lambda)=1]\]

\noindent $\mathsf{S}$ achieves Existential Unforgeability under Chosen Message Attacks (EUF-CMA) if we have $Adv^{\mathsf{EUF-CMA}}_{\mathsf{S},\advA}(\lambda)\le \negl(\lambda)$.
\end{definition}

\begin{figure}[!htb]
\centering
\begin{fboxenv}
\begin{varwidth}[t]{.5\textwidth}
    \begin{onehalfspace}
        \uline{$\mathbf{Game}^{\text{EUF-CMA}}_{\advA,\mathsf{TWS}}(1^\lambda)$:}	\\
        $(\sk, \pk) \gets \Gen( 1^\lambda)$\\
        $\left(m^*,\sigma^*\right) \sample \advA^{\oracle^{\sign}(\cdot)}(\pk)$\\
        $\pcreturn \big( m^* \notin \mathcal{Q}_s ~\land$ \\ \hphantom{2pt}$ \verify(\pk,m^*,\sigma^*)\big)$
    \end{onehalfspace}
\end{varwidth}
\vline ~
\begin{varwidth}[t]{.5\textwidth}
    \begin{onehalfspace}
        \uline{$\oracle^{\sign}(m)$:}\\
        $\sigma \gets \sign\left(\sk,m\right)$ \\
        $\mathcal{Q}_s \gets \mathcal{Q}_s \cup \{m\}$ \\
        $\pcreturn~  \sigma$
    \end{onehalfspace}
\end{varwidth}
\end{fboxenv}
\caption{The EUF-CMA security game.} 
\label{fig:euf_cma}
\end{figure}



\paragraph{Zero-Knowledge Proofs}
Non-interactive zero knowledge (NIZK)~\cite{STOC:GolMicWig87} proof as a strong cryptographic primitive enables a prover to convince a (sceptical) verifier about the truth of a statement without disclosing any additional information in one round of communication. A NIZK can be build in two possible settings: either in Random Oracle Model (ROM) or in the Common Reference String (CRS) model. Next we recall the definition of NIZK proofs in the CRS model and list their main security properties.

\begin{definition}[Non-Interactive Zero-Knowledge Proofs]
	\label{def:nizk}
	Let $\mathcal R$ be an NP-relation, the language $\mathcal L_{\mathcal R}$ can be defined as $\mathcal L_{\mathcal R}=\{x\,|\,\exists ~ w \text{ s.t. } (x,w) \in \mathcal R \}$, where $x$ and $w$ denote public statement and secret witness, respectively. A NIZK, denoted by $\Pi$, for $\mathcal R$ consists of three main PPT algorithms $\Pi=(\setup,\prove,\verify)$ defined as follows:
	\begin{itemize}
	\item $\Pi.\setup(1^\lambda,\mathcal R)\to \crs$: The CRS generation algorithm takes the unary representation of the security parameter $\lambda$ and relation $\mathcal R$ as inputs and returns a set of common reference string $\crs$ as output.
		
	\item $\prove(\crs,x,w) \to \pi$: The prove algorithm takes $\crs$, a public statement $x$ and a secret witness $w$ as inputs, and it then returns a proof $\pi$ as output.
		
    \item $\verify(\crs,x,\pi)\to 0/1$: The verify algorithm takes $\crs$, a public statement $x$ and a proof $\pi$ as input, and it then returns a bit indicating either the acceptance, $1$, or rejection, $0$, as output.
	\end{itemize}
\end{definition}

Informally speaking, a NIZK proof has three main security properties: Completeness, Zero-Knowledge and soundness (extractability), which we formally recall them in below:

\begin{definition}[Completeness] \label{def:completeness}
A NIZK proof, $\Pi$, is called complete, if for all security parameters, $\lambda$, 
and all pairs of valid $(x,w)\in R$ we have,
\[
   \Pr \left[
        \begin{matrix}
            \crs \gets \setup(1^\lambda): \\\ \verify(\crs,x,\prove(\crs,x,w))=1] 
        \end{matrix}
        \right] \ge 1-\negl(\lambda)\; .
\]
\end{definition}

\begin{definition}[Zero-Knowledge]
	\label{def:zk}
	A NIZK proof system, $\Pi$, for a given relation $\mathcal{R}$ and its corresponding language $\mathcal{L}_{\mathcal{R}}$, we define a pair of algorithms $\sim = (\sim_1, \sim_2)$ as the simulator. The simulator operates such that $\sim'(\crs, \mathsf{tpd}, x, w) = \sim_2(\crs, \mathsf{tpd}, x)$ when $(x, w) \in \mathcal{R}$, and $\sim'(\crs, \mathsf{tpd}, x, w) = \bot$ when $(x, w) \notin \mathcal{R}$, where $\mathsf{tpd}$ is a trapdoor. For $b\in\bits$, we define the experiment $\zk^\Pi_{b,\sim}(1^\secpar,\advA)$ in \cref{fig:zk-game}. The associated advantage of an adversary $\advA$ is defined as
	\[
        Adv^\zk_{\Pi,\advA,\sim}(\lambda)\coloneqq 
        \left|
        \begin{matrix}
            \Pr[\zk^\Pi_{0,\sim}(1^\secpar,\advA)=1]- \\
            \Pr[\zk^\Pi_{1,\sim}(1^\secpar,\advA)=1]
        \end{matrix}
        \right| \enspace .
	\]
	A NIZK proof system $\Pi$ achieves perfect and computational zero-knowledge, w.r.t a simulator $\sim=(\sim_1,\sim_2)$, if for all PPT adversaries $\advA$ we have $Adv^\zk_{\Pi,\advA,\sim}(\lambda)=0$, and $Adv^\zk_{\Pi,\advA,\sim}(\lambda)\leq\negl(\lambda)$, respectively.
\end{definition}

\begin{figure}[H]
	\centering
	\begin{fboxenv}
		\begin{varwidth}[t]{.5\textwidth}
			\begin{onehalfspace}
				\uline{$\zk_{0,\sim}^\Pi(1^\lambda,\advA)$\hfill} \\
				$\crs\leftarrow\setup(1^\lambda)$ \\
				$\alpha\leftarrow\advA^{\prove(\crs,\cdot,\cdot)}(\crs)$ \\
				$\pcreturn ~\alpha$
			\end{onehalfspace}
		\end{varwidth}
		\vline
		\hspace{0pt}
		\begin{varwidth}[t]{.5\textwidth}
			\begin{onehalfspace}
				\uline{$\zk_{1,\sim}^\Pi(1^\lambda,\advA)$\hfill} \\
				$(\crs,\mathsf{tpd})\leftarrow\sim_1(1^\lambda)$ \\
				$\alpha\leftarrow\advA^{\sim'(\crs,\mathsf{tpd},\cdot,\cdot)}(\crs)$ \\
				$\pcreturn ~ \alpha$
			\end{onehalfspace}
		\end{varwidth}
	\end{fboxenv}
	\caption{Zero-knowledge security property of a NIZK, $\Pi$.}
	\label{fig:zk-game}
\end{figure}

\begin{definition}[Extractability]
	\label{def:ext}
	A NIZK proof system $\Pi$ for a relation $\mathcal R$ and the language $L$ is called extractable~\cite{EC:CKLM12} if there exists a pair of algorithms $\ext:=(\ext_1,\ext_2)$ called extractors with the following advantage for all PPT adversaries $\advA$:

		\[
  \begin{split}
Adv^{\mathrm{CRS}}_{\Pi,\advA}\coloneqq& |\Pr[\crs\gets\setup(1^\secpar);1\gets\advA(\crs)] - \\& \Pr[(\crs,\state)\gets\ext_1(1^\lambda);1\gets\advA(\crs)]|\; ,
  \end{split}
		\]
	
	and

		\[
		Adv^\extract_{\Pi,\advA}(\secpar)\coloneqq\Pr\left[\begin{aligned}
			&(\crs_\ext,\state_\ext)\gets\ext_1(1^\lambda)\\&
			(x,\pi)\gets\advA(\crs_\ext):\\&
			\verify(\crs_\ext,x,\pi)=1 ~ \land \\& (x,\ext_2(\crs_\ext,\state_\ext,x,\pi))\not \in \mathcal R
		\end{aligned}
		\right] \; .
		\]

A NIZK proof system $\Pi$ is called extractable, w.r.t an extractor $\ext=(\ext_1,\ext_2)$, if $Adv^{\textup{CRS}}_{\Pi,\advA} \leq\negl(\secpar)$ and 	$Adv^\extract_{\Pi,\advA}(\secpar)\leq\negl(\secpar)$. Additionally, we refer to an extractable NIZK proof as a non-interactive zero-knowledge proof of knowledge, or NIZKPoK in short.

\smallskip

 \textbf{Succinctness}. Zero-Knowledge Succinct Non-Interactive Arguments of Knowledge, zkSNARK in short, are NIZKPoK proofs that adhere to succinctness requirements. These proofs maintain communication complexity (proof size) at sublinear levels, and in some cases, the verifier's computational workload remains sublinear, regardless of the size of the witness. In this paper, we primarily concentrate on zkSNARKs, ensuring that the proofs are short and verification cost is low while the mentioned security definitions for NIZK remain applicable for them. 
\end{definition}



\paragraph{zkLogin Completeness}
We define the completeness property for \sysname as follows.

\begin{definition}[zkLogin Completeness]
\sysname achieves completeness if for all $\zkaddr, \iss, T_{exp}, T_{cur}, M$, and sufficiently large security parameter $\lambda$, we have:
\[
\Pr\left[
\begin{matrix}
    \pk \gets \mathsf{zkLoginGen}(1^{\lambda}),
    \\
    \sigma \gets \mathsf{zkLoginSign}(\pk, \zkaddr, \iss, M, T_{exp}),
    \\
    \sigma \neq \bot \land T_{cur} < T_{exp}
    :
    \\
    \mathsf{zkLoginVerify}(\pk, \zkaddr, \iss, M, \sigma, T_{cur})=1
\end{matrix}
\right]\ge 1-\negl(\lambda) \; \cdot
\]
\end{definition}

    \section{Deferred Proofs}
\label[appendix]{sec:zklogin-as-tws}

\zkloginuf*
\begin{proof}

We prove this theorem using a sequence of games.

\textbf{Game 0.} This game is the same as the one defined in Def.~\ref{def:UFCMA_TWS}. 

\textbf{Game 1.} This proceeds as in the real construction, except that, while verifying the forged signature $(\tag^*, M^*, vk_u^*, T^*, \sigma^*_u, \pi^*)$, the challenger also checks whether $vk_u^*$ was used in a previous $\oracle^{\Wit}(\cdot)$ or $\oracle^{\sign}(\cdot)$ responses. If it was not present in any $\oracle^{\Wit}(\cdot)$ responses, but was used in an $\oracle^{\sign}(\cdot)$ response for a different message, then the adversary loses.

This game is indistinguishable from the real protocol, i.e. Game 0, by the EUF-CMA security of $Sig$.

\textbf{Game 2.} This proceeds as in Game 1, except that, while verifying the forged signature $(\tag^*, M^*, vk_u^*, T^*, \sigma^*_u, \pi^*)$, the challenger additionally uses the knowledge extractor for $\Pi$ to extract witness $(\jwt^*, \salt^*, \nonceRand^*)$ 
and the adversary loses if $P_{zk}((\curOPpk, \iss, \zkaddr^*, T^*, vk_u^*), (\jwt^*, \salt^*, \nonceRand^*))$ is false.

This game is indistinguishable from Game 1 by the knowledge-soundness of $\Pi$.

\smallskip
\textbf{Game 3.} This game runs identically to Game 2, except the adversary loses if there was no $\oracle^{\Wit}(\cdot)$ response of the form $(\jwt^*, \cdots)$.

This game is indistinguishable from Game 2 by the EUF-CMA security of $\issueJWT$, guaranteed by the underlying signature scheme of the issuer. 

\smallskip

Given the indistinguishability of the real protocol and Game 3, the unforgeability adversary succeeds in Game 3 with a probability that is negligibly away from that in the real game. Now observe that the Game 3 the adversary never wins. Hence the protocol is an unforgeable \TWS.
\end{proof}

\zkloginpv*

\begin{proof} 
The algorithm $\SimGen(\cdot)$ generates the $zkcrs$ in simulation mode with trapdoor $\td$. For any signing query $\tag = (\curOPpk, \iss, \zkaddr, T)$, the oracle $\oracle^{\simsign}(\cdot)$ generates a fresh key-pair $(vk_u, sk_u) \gets Sig.\Gen(1^\lambda)$, sets $\zkx \gets (\curOPpk, \iss, \zkaddr, T, vk_u)$, computes $\sigma_u \gets Sig.\sign(sk_u, M)$, computes $\pi \gets \Pi.Sim(zkcrs, \td, \zkx)$, and outputs signature $(vk_u, T, \sigma_u, \pi)$.

This game is indistinguishable from real protocol by the ZK property of $\Pi$. Hence the proposed \TWS achieves witness hiding.
\end{proof}




\zkloginsec*

\begin{proof}
We show that we can build an unforgeability adversary $\advA'$ for $\systws$, given a security adversary $\advA$ for \sysname.

$\advA'$ has access to \TWS oracles $\oracle^{\Wit}(\cdot)$ and $\oracle^{\sign}(\cdot)$. 
and access to JWK of OPs.
When $\advA$ requests for $\oracle^{\mathsf{GetWitness}}(\iss, \zkaddr, T)$, $\advA'$ calls $\oracle^{\Wit}(\cdot)$ 
with $(\curOPpk, \iss, \zkaddr, T)$, obtains 
$(\jwt, \salt, r, \vk_u, \sk{u})$ and relays that to $\advA$.
When $\advA$ requests for $\oracle^{\mathsf{zkLoginSign}}(\zkaddr, \iss, m, T)$, $\advA'$ 
first obtains $\curOPpk$ from JWK of OP, 
sets $\tag \gets (\curOPpk, \iss, \zkaddr, T)$, and then calls $\oracle^{\mathsf{Wit}}(\tag)$ obtaining witness $w$. Then it
calls $\oracle^{\sign}(\cdot)$ with $(\tag, pk, w, m)$, obtains $\sigma$ and relays that to $\advA$.

Let $\advA$ return $(\zkaddr^*, \iss^*, m^*, \sigma^*, T^*)$, then $\advA'$ returns $(\iss^*, \zkaddr^*, T^*, m^*, \sigma^*)$.
Now we argue that $\advA'$ wins if $\advA$ wins.
Given the event $E$ there was no $\oracle^{\Wit}(\cdot)$ call with tag $(\curOPpk, \iss^*, \zkaddr^*, T^*)$.
Given the event $F$ there was no $\oracle^{\sign}(\cdot)$ call with tag $(\curOPpk, \iss^*, \zkaddr^*, T^*)$ and message $m^*$. If $\mathsf{zkLoginVerify}$ succeeds then $\systws.\mathsf{verify}$ also succeeds. 
This concludes the reduction.
\end{proof}

\zkloginul*

\begin{proof}
    
We prove unlinkability from WH using a series of games. In all the games, let the adversary send claimsets $\claimset_0, \claimset_1$. 

\textbf{Game 1}. Adversary is given $\zkaddr$ constructed from $\claimset_0$ and access to $\oracle^{\sign}(\cdot)$.

\textbf{Game 2}. Adversary is given $\zkaddr$ constructed from $\claimset_0$ and access to $\oracle^{\mathsf{SimSign}(\cdot)}$. Game 2 is computationally indistinguishable from Game 1 by the Witness Hiding property of $\systws$.


\textbf{Game 3}. Adversary is given $\zkaddr$ constructed from $\claimset_1$ and access to $\oracle^{\mathsf{SimSign}(\cdot)}$. Game 3 is computationally indistinguishable from Game 2 by the Hiding property of the commitment used to construct $\zkaddr$ from the claimset. 


\textbf{Game 4}. Adversary is given $\zkaddr$ constructed from $\claimset_1$ and access to $\oracle^{\sign}(\cdot)$. Game 4 is computationally indistinguishable from Game 3 by the Witness Hiding property of $\systws$.


Note that the transition Game 2 $\to$ Game 3 doesn't work if the adversary had access to $\oracle^{\mathsf{Wit}(\cdot)}$.
\end{proof}





    \section{A Generic \TWS Construction}
\label{sec:nizk-witsig}


We construct a \TWS for the signature verification predicate of a signature scheme $\Sigma = (\keygen, \sign, \verify)$. This instantiation is built upon two primary components: a commitment scheme, $\comm: \{0,1\}^* \to \{0, 1\}^{\lambda}$, and a non-interactive zero-knowledge (NIZK-PoK) scheme $\Pi$. It closely follows the proposed SoK construction in~\cite{C:ChaLys06}.

For a key-pair $(\sk{}, \vk)$ sampled by $\Sigma.\keygen(\lambda)$, we define
\[
    P_{\vk}(t, w) \equiv \Sigma.\verify(\vk, w, t)
\]

Now we describe the \TWS $\Sigma_{sig}$ for the predicate $P_{\vk}$.

\begin{description}
    \item [$\Gen(\lambda)$]:
        \begin{itemize}
            \item Let language $\mathcal L_{\vk} = \{ (t, m, c)\ |\ \exists w, r: c = \comm(t, m, w; r) \land \verify(\vk, w, t) = 1 \}$
            \item Sample $\crs \gets \Pi.\setup(\lambda, {\mathcal L}_{\vk})$.
            \item Output $pk := \crs$
        \end{itemize}
    \smallskip
    \item [$\sign(t, pk, w, M)$]:
        \begin{itemize}
            \item If $P_{\vk}(t, w)$ is false, then output $\bot$.
            \item Sample random $r$.
            \item Compute $c \gets \comm(t, M, w; r)$.
            \item Compute $\pi \gets \Pi.\prove(\crs, (t, M, c), (w, r))$.
            \item Output $\sigma := (c, \pi)$.
        \end{itemize}
    \smallskip
    \item [$\verify(t, pk, M, \sigma)$]:
        \begin{itemize}
        \item Parse $(c, \pi):=\sigma$
        \item Parse $\crs:=pk$
            \item Output $\Pi.\verify(\crs, \pi, (t, M, c))$
        \end{itemize}
\end{description}


\begin{theorem}
    $\Sigma_{sig}$ is an unforgeable tagged witness signature scheme, given that $\Pi$ satisfies knowledge-soundness, the signature scheme is EUF-CMA secure, and $\comm$ is a binding commitment scheme.
\end{theorem}

\begin{proof}
We construct an attack on the unforgeability of the signature scheme used for witness, given an attack on \TWS unforgeability.

The \TWS unforgeability challenger proceeds as in the real construction, except that it uses the knowledge extractor of $\Pi$, and has access to the signing oracle $\sign(\sk{}, \cdot)$. 

For an $\oracle^{\sign}$ query $(t, M)$, the challenger queries the signing oracle with $t$, and receives $w$. Then it samples $r$ and computes $c \gets \comm(t, M, w; r)$ and $\pi \gets \Pi.\prove(\crs, (t, M, c), (w, r))$ and outputs $\sigma = (c, \pi)$. 

For an $\oracle^{\Wit}$ query $t$, the challenger queries the signing oracle with $t$, and receives $w$. 

When it receives a forgery $\sigma^* = (t^*, M^*, c^*, \pi^*)$, such that $\pi^*$ verifies, it extracts witness $w^*, r^*$ such that $c^* = \comm(t^*, M^*, w^*; r^*) \land \verify(\vk, w^*, t^*) = 1$. 
Then it outputs $(t^*, w^*)$ to the signature scheme challenger. 

Due to the knowledge-soundness of $\Pi$, we should have, with high probability, $c^* = \comm(t^*, M^*, w^*; r^*) \land \verify(\vk, w^*, t^*) = 1$.

Due to the unforgeagibility of the signature scheme, we should have that $t^*$ was, with high probability, queried by the challenger to the signature oracle, either to respond to an $\oracle^{\sign}$ query or an $\oracle^{\Wit}$ query. If $t^*$ was queried to $\oracle^{\Wit}$ then the TWS adversary loses. Otherwise, if $(t^*, M^*)$ was queried to $\oracle^{\sign}$, then also the TWS adversary loses. We only have to consider the case that for all $(t^*, M')$ queried to $\oracle^{\sign}$, we have $M' \neq M^*$. Then we have $\comm(t^*, M^*, w^*; r^*) = \comm(t^*, M', w^*; r')$ for some $(M', r') \neq (M^*, r^*)$. Due to the binding property of $\comm$, this only holds with negligible probability.
\end{proof}

\begin{theorem}
    The above NIZK-based construction is a witness hiding \TWS, given that $\Pi$ satisfies ZK and $\comm$ is a hiding commitment scheme.
\end{theorem}

\begin{proof}
We prove this using a sequence of games.

\smallskip
\textbf{Game 1.} In this game, the challenger generates zkcrs in the simulation mode and holds the trapdoor. Note that for some systems like \cite{EC:Groth16} this is identical to the real mode, while for \cite{EC:GroSah08} these are distinct. The proofs are produced using the simulation mode. This game is indistinguishable from the real protocol by the zk property of $\Pi$.

\smallskip
\textbf{Game 2.} In this game, the challenger produces fake signatures using the simulation trapdoor. So it samples $c \gets \comm(0)$ with a fake proof $\pi$ over $(t, m, c)$. This game is indistinguishable from Game 1 by the hiding property of $\comm$. 

\smallskip
Now observe that we can use the Game 2 challenger as the $\simsign$ protocol.
\end{proof}

    \section{Groth16 Ceremony} 
\label{sec:ceremony}

\sysname employs the Groth16 zkSNARK construction as the most efficient zkSNARK to date, to instantiate the zero-knowledge proofs. However, this construction requires a circuit-specific Common Reference String (CRS) setup by a trusted party, we use well-known trust mitigation techniques to relax this trust assumption. We ran a ceremony protocol to generate this CRS which bases its security on the assumed honesty of a single party out of a large number of parties.

The ceremony essentially entailed a cryptographic multi-party computation (MPC) conducted by a diverse group of participants to produce this Common Reference String (CRS). This process followed the MPC protocol for reusable parameters, referred to as MMORPG, as detailed by Bowe, Gabizon, and Miers in~\cite{eprint2017:BoweGM}. The protocol roughly proceeds in 2 phases. The first phase yields a sequence of monomials, which are essentially powers of a secret value $\tau$ in the exponent of a generator of a pairing-friendly elliptic curve. This sequence takes the form of $g, g^{\tau}, g^{\tau^2}, \ldots, g^{\tau^n}$, where $n$ is an upper bound of circuit size and $\tau=\prod_{i=1}^\ell{\tau_i}$ s.t. $\ell$ denotes the total number of contributors. Thereby it enables reducing the trust level to 1 out of the total number of contributors, i.e., $\ell$. As this phase is not specific to any particular circuit, we have adopted the outcome of the Perpetual Powers of Tau\footnote{\href{https://github.com/privacy-scaling-explorations/perpetualpowersoftau}{https://github.com/privacy-scaling-explorations/perpetualpowersoftau}}, which is contributed by a sufficiently large community. However, in order to fully implement the trust minimization process, we must establish a ceremony for the second phase of setup, which is tailored to the \sysname circuit.

In the presence of a coordinator, the MMORPG protocol allows an indefinite number of parties to participate in sequence, without the need of any prior synchronization or ordering. Each party needs to download the output of the previous party, generate randomness of its own, and then layer it on top of the received result, producing its own contribution, which is then relayed to the next party. The protocol guarantees security, if at least one of the participants follows the protocol faithfully, generates strong randomness and discards it reliably.

Since the MPC is sequential, each contributor had to wait until the previous contributor finished in order to receive the previous contribution, follow the MPC steps, and produce their own contribution. Due to this structure, participants waited in a queue while those who joined before them finished. To authenticate participants, each participant received a unique activation code. The activation code was the secret key of a signing key pair, which had a dual purpose: it allowed the coordination server to associate the participant’s email with the contribution, and it verified the contribution with the corresponding public key.


Participants had two ways to contribute: through a browser or a docker. The browser option was the more user-friendly as all parts of the process happened in the browser. The Docker option required Docker setup but was more transparent—the Dockerfile and contributor source code are open-sourced and the whole process is verifiable. The browser option utilized snarkjs\footnote{\href{https://github.com/iden3/snarkjs}{https://github.com/iden3/snarkjs}} while the Docker option utilized a forked version of Gurkan’s implementation\footnote{\href{https://github.com/kobigurk/phase2-bn254}{https://github.com/kobigurk/phase2-bn254}}. This provided software variety so that contributors could choose whichever method they trust most. In addition, participants could generate entropy via entering random text or making random cursor movements.

The \sysname circuit and the ceremony client code were made open source and the links were made available to the participants to review before the ceremony, if they chose to do so. In addition, developer docs and an audit report on the circuit were posted for review. Challenge number 81 was adopted (resulting from 80 community contributions) from perpetual powers of tau in phase 1, which is circuit agnostic. The output of the Drand random beacon was applied to remove bias. After the phase 2 ceremony, the output of the Drand random beacon was applied again to remove bias from contributions.

The final CRS along with the transcript of every participant’s contribution is available in a public repository. Contributors received both the hash of the previous contribution they were working on and the resulting hash after their contribution, displayed on-screen and sent via email. They can compare these hashes with the transcripts publicly available on the ceremony site. In addition, anyone is able to check that the hashes are computed correctly and each contribution is properly incorporated in the finalized parameters.

We also note that various other trust mitigation techniques have also been defined, including subversion-resistant zkSNARKs~\cite{AC:ABLZ17,PKC:Fuchsbauer18} and Multi-Party Computation (MPC)~\cite{SP:BCGTV15}. Additionally, we leave the implementation of \sysname using universal and updatable zkSNARKs such as Plonk~\cite{EPRINT:GabWilCio19} and Sonic~\cite{CCS:MBKM19}, which removes the complex circuit-dependent setups, as an interesting future extension.

    \section{Miscellaneous details}
\label{sec:misc}


\paragraph{Why include $\aud$ in address derivation?}
Including the audience mitigates two concerns.

First, popular OpenID providers like Google and Facebook are used by many to login to hundreds of websites and applications.
So a crucial requirement for \sysname is that a malicious app or website should not be able to steal a user's funds.
Including audience ($\aud$) in the address derivation achieves this.
In other words, if we had derived address instead as $\zkaddr = H(\sub, \iss)$, then a malicious app will be able to steal the user's assets.
To be precise, this would only be a threat for providers employing public identifiers (\cref{sec:oidc}); the pseudonymity offered by pairwise identifiers protects us from attacks of the above style.

Second, even if pairwise identifiers were used, another reason to include the audience $\aud$ is that it ensures that the addresses of two users won't collide.
This is because an OpenID provider might assign the same subject identifier to two different users in two different contexts.

\paragraph{Escaped quote in a JSON key}
We assumed previously that JSON keys do not have escapes. If this does not hold, an attacker can break security, as shown in~\cref{lst:quote-json-key}.

\begin{lstlisting}[language=json, caption={Quote inside a JSON key. This JSON can be parsed in two ways as shown by the start and end index markers.
The key is \texttt{"sub"} in both but the value is different.}, label={lst:quote-json-key}]
{
    "sub": "110463452167303598383",
    ^                             ^  
    "\\"sub": "110463452167303598382",
       ^                             ^
}
\end{lstlisting}

\fi

\vfill

\end{document}